\documentclass[11pt]{article}

\usepackage[
  font=palatino,
  geometry={letterpaper,margin=1in},
  paragraphs=skip,
  spacing=texdefault,
  links=true,
  linkcolors=false,
  theorems=true,
  share counter=true,
  number in=section,
  algorithms=true,
  algorithmsoptions={algo2e,ruled},
  graphics=true,
  graphicspath={figures/},
  probability=alias,
  biblio=biblatex,
  bibliostyle=alphabetic,
  bibliosources={},
  biblatexbackend=biber
]{pomegranate}
\usepackage{graphicx}
\usepackage{nicefrac}
\usepackage{subcaption}

\newif\ifshowcomments
\showcommentsfalse

\usepackage{xcolor}

\DeclareOperator{\sign}
\DeclareOperator{\sgn}
\DeclareOperator{\Top}
\DeclareOperator{\dist}
\DeclareOperator{\err}
\DeclareOperator{\diam}
\DeclareOperator{\Unif}
\DeclareOperator{\Rad}
\DeclareOperator{\Bern}
\DeclareOperator{\Var}
\DeclareOperator{\Cov}
\DeclareOperator{\polylog}

\DeclareColor{\Red}{192, 0, 0} \DeclareOperator{\sign} \DeclareOperator{\Unif}

\newcommand{\tO}{\widetilde O}

\newcommand{\eps}{\varepsilon}

\newcommand{\cF}{\mathcal F}

\newcommand{\cN}{\mathcal N}
\newcommand{\cS}{\mathcal S}
\newcommand{\cT}{\mathcal T}

\newcommand{\xstar}{x^\star}

\newcommand{\xt}{x^{(t)}}
\newcommand{\xtp}{x^{(t+1)}}
\newcommand{\ut}{u^{(t)}}
\newcommand{\utp}{u^{(t+1)}}
\newcommand{\rt}{r^{(t)}}
\newcommand{\rtp}{r^{(t+1)}}

\newcommand{\BIHT}{\textsc{BIHT}}

\newcommand{\RAIC}{\textsc{RAIC}}

\newcommand{\hA}{h_A}
\newcommand{\hAJ}{h_{A;J}}

\crefname{algorithm}{Algorithm}{Algorithms}
\Crefname{algorithm}{Algorithm}{Algorithms}
\crefname{algocf}{Algorithm}{Algorithms}
\Crefname{algocf}{Algorithm}{Algorithms}

\title{On the Role of Normalization in Binary Iterative Hard Thresholding for 1-bit Compressed Sensing}
\author[1]{Arya Mazumdar}
\author[1]{Prateeti Mukherjee}
\affil[1]{UC San Diego, \url{{arya, pmukherjee}@ucsd.edu}}
\date{}

\AtBeginDocument{\let\epsilon\varepsilon}
\begin{document}
\maketitle
\begin{abstract}
Binary Iterative Hard Thresholding (BIHT) is a simple, yet effective, greedy method for recovering a sparse  vector from one-bit sign measurements.  In its original form, BIHT performs a ``gradient-descent'' step, followed by hard thresholding. 

A convergence analysis of this  algorithm was left open in the introductory work of \textcite{jacques2013robust} and has remained unresolved for over a decade, with subsequent sharp analyses studying a {\em normalized variant} instead, that additionally projects every iterate onto the  unit sphere. This paper resolves that gap and characterizes when per-iterate normalization is algorithmically necessary.

In the noiseless setting, we prove a universal, sample-optimal convergence theorem for the original BIHT algorithm. Specifically, with $\widetilde O\parens *{s/\epsilon}$ measurements, a deterministic, computable finite-time iterate has directional error at most $\epsilon$, simultaneously for every $s$-sparse unit vector. This matches the optimal sample dependence achieved by normalized BIHT in prior work. 

Under sign corruptions, we prove a sharp separation. If at most a $\tau$ fraction of signs are flipped adversarially, then \BIHT{}, without per-iterate normalization, still reaches the robust error floor at an early iterate with a matching $\widetilde O\parens *{s/\epsilon}$ sample complexity rate as its normalized variant. This recovery, however, is not stable. We prove a scalar lower bound showing that any nontrivial corruption pattern, even one that involves only one flipped sign together with one clean sign, forces the iterates to oscillate indefinitely. Consequently, no general last-iterate convergence theorem can hold for \BIHT{} under sign corruptions, while its normalized surrogate provably escapes this instance.
\end{abstract}
%\input{abstract}

%\tableofcontents
\section{Introduction}

In one-bit compressed sensing, a high-dimensional signal is observed only through the signs of linear measurements. Given a sensing matrix $A \in \mathbb{R}^{m \times n}$ and an unknown $s$-sparse signal $x^\star\in\Sigma_s^n$, where $\Sigma_s^n \triangleq\set{x\in\R^n \given \norm{x}_0 \le s}$, the vector of $m$ observations is
\begin{align*}
    y = \sign(A x^\star) \in \{\pm 1\}^m,
\end{align*}
one bit per measurement. This model, first introduced in the work of \textcite{Boufounos20081BitCS}, sits at the extreme end of quantization, discarding all magnitude information. This loss of scale is a consequence of the homogeneous, zero-threshold sign map, not the one-bit alphabet itself. For instance, in the one-bit $\Sigma\Delta$ model of \textcite{saab_wang_yilmaz_2018}, each output bit also depends on an internal state carried across measurements, so the resulting bitstream retains enough scale for full recovery. Under the memoryless sign model considered here, positive rescaling of $x^\star$ leaves every observation unchanged. As a result, the only identifiable component is the direction, and we normalize the target to the sparse unit sphere without loss of generality. The objective is to produce, from as few measurements as possible, a directional estimate $\widehat{x}$ that is $\epsilon$-close to $\xstar$ in $\ell_2$ distance. 

If the sensing matrix has standard Gaussian rows, then $m=\widetilde\Theta\parens*{s/\epsilon}$ measurements are necessary and sufficient \cite{jacques2013robust}, while the first computationally tractable convex program achieved recovery with $m=\widetilde O\parens*{s/\epsilon^5}$ measurements \cite{plan2013one}. Under adversarial sign corruptions, the uniform guarantee for a robust convex program deteriorated further to $m=\widetilde O\parens*{s/\epsilon^{12}}$ measurements \cite{plan_vershynin_robust_2013}. The optimal $m=\widetilde O\parens*{s/\epsilon}$ sample dependence was subsequently recovered through a multi-stage sensing-and-decoding scheme in \cite{gopi_onebit_2013}, but their scheme is considerably more involved than a draw of i.i.d. Gaussian entries, with running times scaling as $(s/\epsilon)^5$. Similar recovery guarantees also appear in \cite{flodin_2019}, but only for a restricted class of signals. 
%, while \textcite{gopi_onebit_2013} require expensive, multi-stage decoders, 

\textcite{jacques2013robust} proposed {\em Binary Iterative Hard Thresholding (BIHT)} (\cref{alg:biht}) as a  simple first-order method for this problem, analogous to the iterative hard thresholding algorithm for compressed sensing~\cite{blumensath2009iterative}. Starting from a sparse unit vector, BIHT iterates a sign-mismatch correction, interpretable as a subgradient descent step on a one-sided sign-consistency objective, followed by hard thresholding onto the set of $s$-sparse vectors, with a single normalization applied to the final output. The algorithm is simple to implement, requires no tuning beyond the sparsity level, and has remained the standard first-order baseline for one-bit recovery. A convergence analysis of BIHT was left open in \cite{jacques2013robust}, and has remained unresolved for over a decade.

\begin{Algorithm}[H]
\caption{Binary Iterative Hard Thresholding (\cite{jacques2013robust})}
\label{alg:biht}

\small

\KwIn{$A \in \R^{m \times n}$, $y \in \set{-1,+1}^m$, sparsity $s$, horizon $T$, $x^{(0)} \in \Sigma_s^n \cap \cS^{n-1}$}
\KwOut{$x_{\mathrm{BIHT}}$}

%For $v\in\R^n$, let $\cT_s(v)$ denote the set of vectors obtained by retaining any $s$ largest coordinates of $v$ and zeroing out the rest\;
Let $\cT_s(v)$ denote all $s$-term hard-thresholdings of $v$, with arbitrary tie-breaking

Set $\eta = \sqrt{2\pi}$\;
%Fix arbitrary $x^{(0)} \in \Sigma_s^n \cap \cS^{n-1}$\;

\For{$t = 1,2,\ldots,T$}{
    %$\widetilde{x}^{(t)}
    $v^{(t)} 
    \leftarrow
    x^{(t-1)}
    +
    \frac{\eta}{2m}A^\top
    \parens*{
            y-\sign\parens*{Ax^{(t-1)}}
    }$\;

    Choose any $x^{(t)}
    \in
    %\mathcal{T}_s\parens{\widetilde{x}^{(t)}}$\;
    \mathcal{T}_s\parens{v^{(t)}}$\;

    \Red{$x^{(t)}
    \leftarrow
    x^{(t)}/\norm{x^{(t)}}_2$} \quad \Red{(per-iterate normalization: not executed in BIHT)}\; 
}

\KwRet{$x^{(T)}/\norm{x^{(T)}}_2$}
\end{Algorithm}

%Progress on this original dynamics have been made over the years. 
Early progress toward a convergence theory for \BIHT{} focused largely on one-step analyses. \textcite{jacques2013quantized} showed that the \emph{first} iterate of \BIHT{} achieves error $\epsilon$ from $m = \tO\parens*{s/\epsilon^2}$ measurements; similar one-step guarantees appear in \textcite{plan_verhsynin_yudovina}. Beyond the first iterate, however, these analyses do not track the \BIHT{} trajectory, and their dependence on $\epsilon$ is quadratically suboptimal. \textcite{liu_li_shen} then studied the full trajectory, interpreting BIHT as a projected subgradient method under sparsity constraints, and proved that the subsequent iterates remain bounded for summable step-sizes. This yields stability of the trajectory, but no reduction of the error beyond the first iteration. A convergence theorem for the algorithm of \cite{jacques2013robust}, one in which the iterates provably \emph{improve}, remained out of reach.

Sharp guarantees were finally obtained by changing the algorithm. \textcite{friedlander_nbiht} and \textcite{Matsumoto_noiseless_2024, Matsumoto_robust_2023, pmlr-v291-matsumoto25a} analyze a \emph{normalized} variant of \BIHT{} that additionally projects every iterate onto the  unit sphere. Normalizing the iterates allows one to translate a so-called {\em restricted approximate invertibility condition (RAIC)} of Gaussian random matrices to show contraction of the estimation error at every step. For the normalized variant of \BIHT{} with Gaussian sensing matrices, the theory is, by now, essentially complete: \cite{Matsumoto_noiseless_2024} proves universal recovery with the optimal $\tO\parens *{s/\epsilon}$ measurements, and \cite{Matsumoto_robust_2023} extends the guarantee to adversarial sign corruptions with a matching robust error floor. In each of these analyses, per-iterate normalization is crucial: since RAIC holds uniformly for any pair of sparse unit vectors, contraction is guaranteed whenever the iterates are on the unit sphere. The original \BIHT{} algorithm offers no such control over the iterates, and its analysis has remained open. 

This paper resolves that open question and, in doing so, characterizes the role of per-iterate normalization in fixed-step BIHT: unnecessary for sample-optimal finite-time recovery, but without it, we cannot have a uniform all-time guarantee under adversarial corruptions.

\subsection{Our Results}
\label{sec:intro-results}
 
\paragraph{Normalization is (almost) unnecessary without noise.} Our first result is a universal, sample-optimal convergence theorem for the original, normalization-free BIHT algorithm under noiseless measurements.
 
\begin{theorem}[Informal version of \cref{thm:noiseless}] \label{thm:intro-noiseless}
Let $0 < \epsilon \le 1/64$ and let $A\in \R^{m\times n}$ have independent standard Gaussian entries. If $m = \widetilde{O}(s/\epsilon)$, then with high probability over the draw of $A$, the following holds simultaneously for every $x^\star \in \Sigma^n_s \cap \cS^{n-1}$ and every sparse unit initialization: BIHT, without per-iterate normalization, run for $T = O\parens{\log\parens{1/\epsilon}}$ iterations on $y = \sign\parens{Ax^\star}$, returns an estimate $\widehat{x}_{\mathrm{BIHT}}$ with $\norm{\widehat{x}_{\mathrm{BIHT}} - x^\star}_2 \le \epsilon$.
\end{theorem}
The sample complexity matches the optimal dependence achieved by normalized BIHT \cite{Matsumoto_noiseless_2024}, and the guarantee is of the same universal form: a single draw of the sensing matrix serves all sparse targets simultaneously. Since $\eps$ is known, the horizon $T$ is computable. Furthermore, the guarantee does not require stopping exactly at $T$, and the iterates remain $\eps$-accurate for $\Omega(1/\eps^2)$ further iterations simultaneously for every target. As $t\to\infty$, our guarantee for the original algorithm is weaker than that of normalized BIHT, and, as we prove in \cref{thm:noiseless-late-failure}, necessarily so. With $\widetilde{O}\parens{s/\min\{\eps,\min_{i\in\supp(\xstar)}\abs{x^\star_i}\}}$ measurements, BIHT iterates also remain $\eps$-accurate for every $t\ge T$ (\cref{thm:noiseless-all-time}). This matches the $\widetilde{O}(s/\eps)$ measurements of normalized BIHT when every nonzero coordinate of $\xstar$ is at least $\eps$ in magnitude, and is therefore weaker than the corresponding $t\to\infty$ guarantee of \cite{Matsumoto_noiseless_2024}. While RAIC was the property that enabled the contraction lemmas for normalized BIHT in \cite{Matsumoto_noiseless_2024}, we prove that RAIC alone cannot remove the dependence on the smallest target coordinate in the analysis of normalization-free BIHT as $t\to\infty$. In particular, we construct a matrix satisfying RAIC, and a target with one very small entry, for which normalization-free BIHT remains accurate well past the window above and then loses this accuracy at a late iteration (\cref{thm:noiseless-late-failure}).

\paragraph{Under corruption, the optimum cannot be held\ldots} Our second result concerns  sign corruptions, where the observed labels $y$ may differ from $\sign(Ax^\star)$ in up to a $\tau$ fraction of coordinates. We find that a single corrupted measurement already negates the feasibility of last-iterate convergence for the fixed-step dynamics, even in dimension one. 
 
\begin{theorem}[Informal version of \cref{thm:adversarial-lower-bound}]
\label{thm:intro-lower}
Let $n = s = 1$ and let $A$ have independent standard Gaussian entries. For every corruption pattern that flips at least one sign and leaves at least one sign clean, and every fixed nonzero initialization independent of $A$, the normalization-free fixed-step iterates almost surely change sign infinitely often. Consequently, with probability 1,
\begin{equation*}
    \limsup_{t \to \infty} \norm*{ \frac{x^{(t)}}{\abs{x^{(t)}}} - x^\star }_2 = 2.
\end{equation*}
\end{theorem}
That is, the directional error returns to its maximum possible value infinitely often, even when a single corrupted bit lies within otherwise clean measurements. Since the guarantee being ruled out is uniform over dimensions, sparsity levels, targets, and corruption patterns, this single scalar instance suffices to demonstrate impossibility. On the same instance, if there are more clean bits than flipped ones, normalized BIHT recovers the correct sign from the first step onward, except with exponentially small probability (\cref{rem:normalization-escapes-lb-instance}).

\paragraph{\ldots but it is reached, and one can certify when.} In light of \cref{thm:intro-lower}, a hitting-time guarantee is the strongest form of convergence available to the fixed-step dynamics. Our third result shows that BIHT, without per-iterate normalization, achieves it at the optimal accuracy scale. Our results are stated for the robust error floor
\begin{equation*}
    \gamma(\epsilon, \tau) \triangleq \epsilon + c_1\sqrt{\epsilon\tau} + c_2\, \tau\sqrt{\log(2e/\tau)},
\end{equation*}
which is the same floor attained by normalized BIHT in this model \cite{Matsumoto_robust_2023}.
 
\begin{theorem}[Informal version of \cref{thm:adversarial-positive}] \label{thm:intro-hitting}
Under the same measurement budget as \cref{thm:intro-noiseless},  the following holds simultaneously with high probability for every sparse unit target, every admissible corruption pattern, and  every sparse unit initialization: 
%, and every sequence of hard thresholding choices
 if $\gamma\parens{\epsilon,\tau} \le 1/64$, then some time $t_{\mathrm{hit}} = O\parens{\log\parens{1/\gamma\parens{\epsilon,\tau}}}$ of normalization-free BIHT satisfies
\begin{align*}
\norm*{\frac{x^{(t_{\mathrm{hit}})}}{\norm{x^{(t_{\mathrm{hit}})}}_2} - \xstar}_2 \le \gamma\parens{\epsilon, \tau},
\end{align*}
and every iterate up to $x^{(t_{\mathrm{hit}})}$ is nonzero with norm within $1/64$ of one.
\end{theorem}
The trajectory is generated without knowledge of $\tau$, but the certificate identifying the relevant accuracy scale is $\tau$-dependent. \cref{thm:intro-hitting} is therefore a hitting-time guarantee, and, by \cref{thm:intro-lower}, necessarily so. When an apriori corruption budget $\tau \le \tau_0$ is available, the guarantee upgrades: running normalization-free BIHT for a fixed horizon $T_0 = O(\log(1/\gamma\parens{\epsilon,\tau_0}))$ returns an estimate within the worst-case floor $\gamma\parens{\epsilon, \tau_0}$ (\cref{cor:adversarial-deterministic-stopping}), and the accuracy persists for a window of $\Omega(1/\gamma(\epsilon,\tau_0))$ subsequent iterations, so the stopping time is stable to timing errors (\cref{cor:adversarial-stable-window}). Together, Theorems~\ref{thm:intro-lower} and \ref{thm:intro-hitting} help explain an interesting empirical phenomenon that appears in \cite{pmlr-v291-matsumoto25a} (Fig 1.). They observed a related instability of normalization-free \BIHT{} for noisy binary generalized linear models, where the error curve first drops to a robust floor and then oscillates around a constant error. Our experiments in \cref{fig:numerics-corruptions} exhibit the corresponding early-recovery and late-instability behavior under random sign corruptions.

\section{Preliminaries}
\subsection{Notation}\label{sec:notation}

For $s\in[n]$, $\Sigma_s^n\triangleq\set{x\in\R^n\given \norm{x}_0\le s}$ denotes the set of $s$-sparse vectors in $\R^n$. We write $\cS^{n-1}$ for the Euclidean unit sphere and use $\xstar\in \Sigma_s^n \cap \cS^{n-1}$ for the unknown sparse unit vector. The sign function is defined by $\sign(t)=1$ for $t\ge0$ and $\sign(t)=-1$ otherwise, and is applied coordinate-wise to vectors. For $U\subseteq[n]$, $\cT_U$ denotes coordinate restriction to the set $U$. For $v\in\R^n$, $\cT_s(v)$ denotes the set of Euclidean projections of $v$ onto $\Sigma_s^n$. Equivalently, $\cT_s(v)$ is the set of all hard-thresholded vectors obtained by retaining any $s$ largest coordinates of $v$ in magnitude and zeroing out the rest, with arbitrary tie-breaking. We write $x\in\cT_s(v)$ when $x$ is any such valid hard-thresholded output. Throughout, $\eta=\sqrt{2\pi}$ denotes the fixed BIHT step-size and $A\in \R^{m\times n}$ denotes the sensing matrix. We use $\gtrsim$ to denote $\ge$ up to universal constants.

Under noiseless measurements, $y = \sign(Ax^\star)$. We use the following notation for the correction term of \cref{alg:biht}
\begin{align}\label{eq:hA-def}
\hA(p,q) \triangleq \frac{\eta}{m}A^\top\cdot\frac{1}{2} \parens*{\sign(Ap)-\sign(Aq)}, 
\end{align}
for $p,q \in \R^n.$
For a set $J\subseteq[n]$, we use $\hAJ$ to denote the correction term restricted to the supports of the variables it applies to, and the coordinates in $J$. Formally,
\begin{align*}
\hAJ(p,q) \triangleq \cT_{\supp(p)\cup\supp(q)\cup J} \parens*{\hA(p,q)} 
\end{align*}
The role of the restriction is to match the coordinates that can affect a hard-thresholded update. In our proofs, the set $J$ is usually selected by the next hard-thresholding step, and is therefore data-dependent.

Under measurements corrupted by adversarial noise, the observation vector $y\in\braces *{\pm 1}^m$ satisfies $d_H(y, \sign(Ax^\star))\le \tau m$, where $d_H\parens *{\cdot}$ denotes the Hamming distance. For a fixed matrix $A$ and corruption level $\tau\in[0,1]$, let
\begin{align*}
\cF_A(\tau)\triangleq\set*{f:\R^n\to\set{\pm1}^m\given d_{\mathrm H}(f(z),\sign(Az))\le \tau m\text{ for every }z\in\R^n}. 
\end{align*}
The set $\cF_A(\tau)$ consists of all possible ways an adversary can corrupt the clean measurements, given they are only allowed to change at most a $\tau$ fraction of the data. We will need to then modify the correction term in \cref{eq:hA-def} to measure the mismatch between the corrupted measurements and the current iterate signs. Formally, for every corruption pattern $f\in\cF_A(\tau)$, define
\begin{align*}
h_{f;A}(p,q)
\triangleq
\frac{\eta}{m}A^\top\cdot\frac{1}{2}
\parens*{f(p)-\sign(Aq)}
\end{align*}
and, for $J\subseteq[n]$,
\begin{align*}
h_{f;A;J}(p,q)
\triangleq
\cT_{\supp(p)\cup\supp(q)\cup J}
\parens*{h_{f;A}(p,q)}.
\end{align*}
As before, the support restriction ensures that the analysis matches the coordinates that can meaningfully affect the hard-thresholded updates, and involves data-dependent coordinates. 
\subsection{Deterministic Tools}\label{sec:deterministic-tools}

We first record two elementary facts about hard thresholding.

\begin{lemma}[Hard thresholding is a two-approximate projector] \label{lem:2-approx-projector}
Let $v\in\R^n$, $w\in\cT_s(v)$, and let $z\in\Sigma_s^n$. Then,
\begin{align*}
\norm{w-z}_2
\le
2\norm{v-z}_2
\end{align*}
\end{lemma}

\begin{proof}
Since $w$ is a closest $s$-sparse vector to $v$, we have $\norm{v-w}_2\le\norm{v-z}_2$. The claim follows by the triangle inequality.
\end{proof}

\begin{lemma}[Restricting to a support containing the selected output]
\label{lem:support-restriction}
Let $v\in\R^n$, let $z\in\cT_s(v)$, and let $U\subseteq[n]$ satisfy $\supp(z)\subseteq U$. Then $z\in\cT_s(\cT_U(v))$.
\end{lemma}

\begin{proof}
Let $I=\supp(z)$. Since $z$ is a valid hard-thresholded output of $v$, the support $I$ captures at least as much squared mass of $v$ as any other support of size at most $s$. Restricting to a set $U$ that contains $I$ leaves the mass on $I$ unchanged and cannot improve any support outside $U$. Thus $I$ remains an optimal hard-thresholding support after restriction.
\end{proof}

The following fact is specific to one-bit measurements.

\begin{lemma}[The sign correction only sees direction]
\label{lem:sign-norm-invariant}
Let $z\neq0$ and set $u=z/\norm{z}_2$. Then, for every $p\in\R^n$, $\hA(p,z)=\hA(p,u)$. Moreover, for every $J\subseteq[n]$, $\hAJ(p,z)=\hAJ(p,u)$. The same statement holds for $h_{f;A}$ and $h_{f;A;J}$ for every $f\in\cF_A(\tau)$.
\end{lemma}

\begin{proof}%[Proof]
Since $z=\norm{z}_2u$ with $\norm{z}_2>0$, positive rescaling preserves every coordinate of $\sign(Az)$, including zero coordinates. Furthermore, $z$ and $u$ have the same support.
\end{proof}

Finally, we record the following elementary geometric fact: a non-unit vector that is close to a positive multiple of $x\in \cS^{n-1}$ is close to $x$ after normalization.

\begin{lemma}[Normalization around a radial target] \label{lem:normalization-around-radial-target}
Let $x\in \cS^{n-1}$, let $\alpha>0$, and let $z\in\R^n$ satisfy $\norm{z-\alpha x}_2\le \alpha/2$. Then $z\neq0$, and
\begin{align*}
\norm*{\frac{z}{\norm{z}_2}-x}_2
\le
\frac{4}{\alpha}\norm{z-\alpha x}_2
\end{align*}
\end{lemma}

\begin{proof}
Write $z=\alpha x+\nu$. The assumption gives $\norm{z}_2\ge\alpha/2$. Comparing $z/\norm{z}_2$ first to $\alpha x/\norm{z}_2$, and then using $\abs{\norm{z}_2-\alpha}\le\norm{\nu}_2$, gives the claim.
\end{proof}

\subsection{Random Matrix Tools}

\begin{definition}[Restricted approximate invertibility condition] \label{def:raic-noiseless}
Let $c_1,c_2,\delta>0$. We say that $A$ satisfies $(s,n,\delta,c_1,c_2)$-\RAIC{} if, for every $p,q\in\Sigma_s^n\cap S^{n-1}$ and every $J\subseteq[n]$ with $\card J\le s$, 
\begin{align*}
\norm{(p-q)-\hAJ(p,q)}_2 \le c_1\sqrt{\delta\norm{p-q}_2} + c_2\delta 
\end{align*} 
\end{definition}

\begin{fact}[Gaussian matrices satisfy \RAIC; \textcite{Matsumoto_noiseless_2024}]
\label{fact:gaussian-matrices-satisfy-raic}
There exist absolute constants $c_1,c_2>0$ such that the following holds. Let $1\le s\le n$, $0<\delta<1$, and $0<\rho<1$. Let $A\in\R^{m\times n}$ have independent $\cN(0,I_n)$ rows. If
\begin{align*}
m
\gtrsim
\frac{s}{\delta}\log\parens*{\frac{en}{s}}
\sqrt{\log\parens*{\frac{2e}{\delta}}}
+
\frac{s}{\delta}\log^{3/2}\parens*{\frac{2e}{\delta}}
+
\frac{1}{\delta}\log\parens*{\frac{1}{\rho}}
\sqrt{\log\parens*{\frac{2e}{\delta}}}
\end{align*}
then $A$ satisfies $(s,n,\delta,c_1,c_2)$-\RAIC{} with probability at least $1-\rho$.
\end{fact}

\begin{fact}[Gaussian matrices satisfy adversarial \RAIC{}; \cite{Matsumoto_robust_2023}, with the sample complexity corrected via \cite{Matsumoto_noiseless_2024}]
\label{fact:gaussian-matrices-satisfy-adv-raic}
There exist absolute constants $b_1, b_2, b_3, b_4 > 0$ such that the following holds. Let $1 \le s \le n$, $0 < \delta < 1$, $0 < \rho < 1$, and $0 \le \tau \le 1$. Let $A \in \mathbb{R}^{m \times n}$ have independent $\mathcal{N}(0, I_n)$ rows. If
\begin{align*}
	m
	\gtrsim
	\frac{s}{\delta} \log\parens*{\frac{en}{s}} \sqrt{\log\parens*{\frac{2e}{\delta}}}
	+ \frac{s}{\delta} \log^{3/2}\parens*{\frac{2e}{\delta}}
	+ \frac{1}{\delta} \log\parens*{\frac{1}{\rho}} \sqrt{\log\parens*{\frac{2e}{\delta}}}
\end{align*}
then, with probability at least $1 - \rho$, the following holds simultaneously for every $f \in \mathcal{F}_A(\tau)$, every $p, q \in \Sigma_s^n \cap S^{n-1}$, and every $J \subseteq [n]$ with $\abs{J} \le s$:
\begin{align*}
	\norm*{(p - q) - h_{f;A;J}(p, q)}_2
	\le
	b_1 \sqrt{\delta \norm{p - q}_2}
	+ b_2 \delta
	+ b_3 \sqrt{\delta \tau}
	+ b_4 \tau \sqrt{\log\parens*{\frac{2e}{\tau}}},
\end{align*}
with the convention that $\tau\sqrt{\log(2e/\tau)} = 0$ when $\tau = 0$.
\end{fact}

\section{Overview of Techniques}

\paragraph{Population geometry.} Our proofs crucially depend on the geometric structure of the problem. We will first isolate the behavior of the iterates at the population level. Consider noiseless measurements $y=\sign(Ax^\star)$ where $A$ has standard Gaussian rows $\parens*{a_1, \cdots , a_m}$, and a direction $u\in \cS^{n-1}$. By rotational invariance of the Gaussian distribution,
\begin{align*}
    \E_{a\sim \cN(0,I_n)}{a\operatorname{sign}(\dotprod{a,u})} = \sqrt{\frac{2}{\pi}}u
\end{align*}
%Any $a\in A$ may be decomposed with respect to its alignment with a direction $u\in\cS^{n-1}$
%\begin{align*}
%	a = g u + b
%\end{align*}
%where $g=\dotprod{a,u}\sim \cN(0,1)$ and $b\perp u$ is independent of $g$, with $\E{b}=0$. Since the value of $\sign(0)$ is irrelevant for a Gaussian random variable,
%\begin{align*}
%	\E{a\operatorname{sign}(\dotprod{a,u})}
%	=
%	%\E{(g u+b)\operatorname{sign}(g)} \\
	%&=
%	\E{\abs{g}}u+\E{b}\E{\operatorname{sign}(g)} %\\
%	%&
%    =
%	\sqrt{\frac{2}{\pi}}u
%\end{align*}
Consequently, for unit vectors $p,q\in\cS^{n-1}$,
\begin{align}\label{eq:population-correction}
	\E *{\frac{\sqrt{2\pi}}{2}a\parens{\operatorname{sign}(\dotprod{a,p})-\operatorname{sign}(\dotprod{a,q})}}
	=
	p-q.
\end{align}
The random vector inside the expectation is precisely the single-measurement contribution to the BIHT correction $h_A(p,q)$ in \cref{eq:hA-def}; indeed, for fixed $p,q$, the empirical correction $h_A(p,q)$ is the average of these contributions over the rows of $A$.
%Note that this is exactly the per-iterate correction term in Algorithm~\ref{alg:biht}, with $p=x^\star, q = \xt$, and constant step-size $\eta = \sqrt{2\pi}$. 
%For $p = \xstar$ and $q = \xt$, the expectation is over the correction term in \cref{alg:biht}. Substituting the notation from \cref{eq:hA-def} gives
%\begin{align*}
%    \E *{h_A(\xstar,\xt)} = \xstar-\xt
	%:= \frac{\sqrt{2\pi}}{m} A^\top \cdot \frac{1}{2}
	%\parens *{\operatorname{sign}(A p)-\operatorname{sign}(A q)}
%\end{align*}
%for the correction term. 
%For normalized BIHT, the population identity aligns perfectly with the update. 

Consider the corresponding population dynamics, in which the empirical correction is replaced by its expectation over a row of $A$ and the current state is held fixed. If the fixed state $z$ is a unit vector, then \cref{eq:population-correction} applies to the pair $(\xstar,z)$, and the population pre-threshold update is
%$\ut = \nicefrac{\xt}{\norm{\xt}_2}$, then, prior to thresholding, the next iterate is
\begin{align*}
	z+\E{h_A(x^\star,z)}
	=
	z+(x^\star-z)
	=
	x^\star
\end{align*}
We often refer to this pre-threshold vector, and its empirical counterpart, as the dense iterate. Observe that, prior to thresholding, the update is exactly aligned with the target when $z\in \cS^{n-1}$. Therefore, keeping every iterate on the unit sphere makes the population dynamics appear trivial: each correction points exactly from the current iterate to the target. Hard thresholding does not disrupt this geometry. Indeed, if the dense vector is close to $\xstar$, then every hard-thresholded output is also close to $\xstar$ up to constant factors, since $\xstar\in\Sigma_s^n$ is a valid comparison point for projection onto $\Sigma_s^n$ (see \cref{lem:2-approx-projector}).

Without normalization, the population dynamics does not enjoy this property. If $z\neq 0$, we write
%For the remainder of this exposition, we will isolate the iterate's direction from it's norm. In particular, for iterate $\xt$ at time $t$, we write
\begin{align*}
	z= r u %,
	%\qquad
	%\rt=\norm{\xt}_2,
	%\qquad
	%\ut=\frac{\xt}{\norm{\xt}_2}
\end{align*}
where the norm $r = \norm{z}_2$ is not necessarily unit, and $u = z/\norm{z}_2$ represents the direction. Since the sign function itself is norm invariant, so is the signed correction term, and 
\begin{align*}
	h_A(x^\star,z)=h_A(x^\star,u)
\end{align*}
Substituting the above in \cref{eq:population-correction} implies that the population correction $\E{h_A\parens *{\xstar, z}}$ still estimates $\xstar - u$, and not $x^\star-z$. Without per-iterate normalization, the dense next iterate is therefore
%population update is therefore
\begin{align}\label{eq:unnormalized-old-pop-error}
    z+\E{h_A(x^\star,z)}
	=
	ru +x^\star-u
	=
	x^\star+(r-1)u
\end{align}
If the above dense iterate is compared against $x^\star$, the term $(r-1)u$ appears as an additional, uncontrolled radial error that the normalized variant avoids by forcing $r= 1$.

We argue that this comparison is the wrong one. Without normalization, the update should not be compared to $x^\star$ but to $r x^\star$ instead, i.e., the vector along the target direction, but with the current norm. Rewriting the unnormalized population update as
\begin{align}\label{eq:unnormalized-new-pop-error}
	r {u}+x^\star-u
	&=
	r x^\star+(r-1)({u}-x^\star)
\end{align}
makes this phenomenon apparent. When compared against $r x^\star$, the radial error no longer appears as the additive error in \cref{eq:unnormalized-old-pop-error}. Together with the directional component, it gives the combined per-iterate error $(r-1)({u}-x^\star)$. This is indeed the right geometry for one-bit recovery, since an error that is purely radial should not meaningfully hurt the final answer. If the current state has the right direction, with $u = x^\star$, then the vector $z=r x^\star$ is already perfectly recovered after the final normalization, for every $r>0$. An analysis that insists on controlling $\norm{z - \xstar}_2$ penalizes the algorithm for failing to learn a scale that the measurements never contained in the first place.  

Moreover, since the radial discrepancy only appears as $\abs{r-1}\norm{u-\xstar}_2$, as long as the norm remains in a constant-sized window around $1$, the population error around the moving reference is controlled by the current directional error. As a result, our proof makes no attempt to show that the norms of the iterates converge to one. Given the extreme quantization of the measurements, such a statement would be unnatural: sign measurements carry no radial information, and the iterates have no apparent reason to prefer unit-scale representatives of the target direction. What we need is strictly weaker: the norm must not drift out of this window during the finitely many iterations needed to recover the direction.

Returning to the empirical trajectory, at time $t$ the moving reference is $\rt\xstar$, where $\rt = \norm{\xt}_2$. Since $\rt\xstar$ is $s$-sparse, it is a valid comparison point for projection onto $\Sigma_s^n$. Thus, by \cref{lem:2-approx-projector}, hard thresholding preserves closeness of the dense iterate to $\rt\xstar$ up to constant factors. Control of this distance is enough to control both the directional error and the radial mismatch. Indeed, since $\xstar$ has unit norm, $\norm{\rt\xstar}_2=\rt$, and by triangle inequality
\begin{align*}
    \abs*{\rtp-\rt}
    =
    \abs*{\norm{\xtp}_2-\norm{\rt\xstar}_2}
    \le
    \norm{\xtp-\rt\xstar}_2 .
\end{align*}
Consequently, the proof reduces to a coupled induction: the directional
error contracts, while the radius moves by at most the same one-step error.

\paragraph{Finite-sample replacement and \RAIC{}. } Thus far, our discussion has relied on population estimates. To turn these insights into a uniform finite-sample statement, we use the restricted approximate invertibility condition of Gaussian random matrices (\cref{fact:gaussian-matrices-satisfy-raic}), which asserts that the empirical sign correction behaves like its expectation, uniformly over sparse directions and over the auxiliary supports that may arise from hard thresholding. The latter uniformity is important because the support selected by the next thresholding step is chosen adaptively from the data. Once that is fixed, the realized thresholded iterate can be analyzed by restricting the dense update to the target support, the current support, and the newly selected coordinates. The newly selected coordinates can then be treated as the auxiliary support in \RAIC{}, and the uniform estimate applies to the realized update via \cref{lem:support-restriction}. 

Given the moving-reference decomposition, it is natural to ask for a different route around normalization: extend \RAIC{} itself to non-unit iterates, so that the empirical correction directly estimates $\xstar - \xt$ rather than $\xstar - \ut$. We argue that such an extension is impossible. Taking $q=cp$ for any $c>0, c\neq 1$, gives $h_{A;J}(p,q) = 0$ while $\norm{p-q}_2 = \abs*{1-c}\norm{p}_2$ is arbitrary. Consequently, no inequality of the form
\begin{align*}
	\norm{(p - q) - h_{A;J}(p, q)}_2
	\le
	a_1\sqrt{\delta \norm{p - q}_2} + a_2\delta
\end{align*}
can hold off the unit sphere, for any sensing matrix and any set of constants. %The barrier is information-theoretic: sign measurements carry no radial information, so no function can estimate a radial displacement. 

The proof of \RAIC{} in \textcite{Matsumoto_noiseless_2024, Matsumoto_robust_2023} reflects this at every step: each of its quantitative ingredients is calibrated in the angular metric, which agrees with the Euclidean metric only on the sphere. The number of measurements separating a pair $(u, v)$ concentrates around $\nicefrac{\theta_{u,v}\, m}{\pi}$, a function of the angle alone, and the deviation terms are sized by the spherical distance.  The union bounds run over covers of the compact set $\Sigma_s^n \cap\cS^{n-1}$, and, in the small-distance regime, over the sign patterns realized within an angular cap -- a count to which radial perturbations contribute nothing, since positive rescaling creates no new sign patterns. In sum, \RAIC{} is a statement about directions because directions are all that the measurements contain.

In prior analyses, projecting every iterate onto the unit sphere has been a means to keep the trajectory in the domain of \RAIC{}, and this dependence is not special to \BIHT{}. For instance, \textcite{chen_yuan_quantized_2024} generalize \RAIC{} to star-shaped signal sets, and analyze projected gradient descent for dithered and multi-bit quantization, while \textcite{chen_ding_xia_yuan_2025} position \RAIC{} as the nonlinear analogue of the restricted isometry property \cite{candes_tao_2005}. In both works, once the signal norm is unrecoverable, every iterate is pulled back to the unit sphere at each step, for the same reason as it would be for us: to keep it in the domain where RAIC holds. Our proof instead keeps this domain intact and moves everything else. In fact, the same invariance that makes the off-sphere extension impossible is what makes it unnecessary. Since $h_A\parens{\xstar, \xt} = h_A\parens{\xstar, \ut}$, every invocation of \RAIC{} in our analysis is applied to the pair $\parens{\xstar, \ut}$ of unit vectors, and the norm of the iterate is absorbed into the comparison point $\rt \xstar$ rather than into the domain of the probabilistic estimate. The resulting radial error $\parens{\rt - 1}\parens{\ut - \xstar}$ is a quantity \RAIC{} could never control, and for the finitely many iterations that recovery takes, it never needs to, since it is controlled by the algorithm's own bookkeeping. 
%\prateeti{TODO: add a lil smth about the horizon being important here}

\paragraph{Scalar obstruction under corruptions. } Under sign corruptions, the observed labels differ from $\sign(A\xstar)$ in at most a $\tau$ fraction of coordinates, and the behavior of the unnormalized trajectory splits into two phases. \textcite{pmlr-v291-matsumoto25a} observe in numerical simulations that, without normalization, the \BIHT{} error curve first drops to a robust error floor and then begins to oscillate around a constant error. We prove that this oscillation is unavoidable for fixed step-sizes with a remarkably simple one-dimensional construction. 

For $n=s=1$, hard-thresholding is the identity. For a fixed step-size, a flipped set of bits creates two constant pushes: one when the current iterate is positive, and another when it is negative. Concretely, the update reduces to a map of the form 
\begin{align*}
	\xtp
	=
	\begin{cases}
		\xt-\alpha, & \xt>0 \\
		\xt+\gamma, & \xt<0
	\end{cases}
\end{align*}
for positive constants $\alpha$ and $\gamma$ determined by the total weights of corrupted and clean measurements, respectively. This explains the oscillation; a positive iterate is pushed toward the negative side, and vice versa. Unless the trajectory lands exactly at zero, which is a zero probability event for our purposes, it crosses the origin indefinitely. Perhaps surprisingly, this scalar instance persists even if a single measurement is flipped, causing the unnormalized iterates to oscillate even if only one corrupted bit lies within largely clean measurements. 

\paragraph{Early recovery before oscillations. } The lower bound precludes last-iterate convergence, but the simulations point to the existence of an early iterate that reaches the robust error floor. To derive the positive counterpart, we observe that the deterministic skeleton survives corruptions. The moving-reference decomposition is unchanged, and the right point of comparison is still $\rt\xstar$. What changes is the accuracy with which the empirical correction approximates its population counterpart, which is captured by the restricted approximate invertibility condition of Gaussian matrices under adversarial perturbations (see \cref{fact:gaussian-matrices-satisfy-adv-raic}). The noiseless error terms in the one-step estimate are now replaced by the same finite-sample counterparts, with additional corruption-dependent contributions, $\sqrt{\epsilon\tau}$ and $\tau\sqrt{\log\parens*{\nicefrac{2e}{\tau}}}$, that quantify the price of allowing the corruption set to be chosen adversarially. The same coupled induction as before contracts until it reaches the robust error floor $\epsilon + O\parens*{\sqrt{\epsilon\tau}+\tau\sqrt{\log\parens*{\nicefrac{2e}{\tau}}}}$, and proves that the unnormalized trajectory reaches the same adversarial accuracy scale as the normalized surrogate at an early iterate. 

The trajectory itself is $\tau$-agnostic, but the certificate identifying the relevant accuracy scale depends on $\tau$, which makes \cref{thm:adversarial-positive} a $\tau$-dependent hitting-time guarantee. Without knowledge of $\tau$, or at least an upper bound on it, the theorem does not identify which iterate should be returned. With an a-priori budget $\tau \le \tau_0$, the hitting-time guarantee naturally upgrades to a deterministic stopping rule (\cref{cor:adversarial-deterministic-stopping}), and the recovered accuracy is stable over a window of length $\Omega(1/\gamma(\epsilon,\tau_0))$ after the stopping time (\cref{cor:adversarial-stable-window}).

\section{Main Results}

First, we state and prove that BIHT, as proposed in \textcite{jacques2013robust}, converges with optimal number of noiseless Gaussian measurements. 

\subsection{Noiseless Measurements}
%\prateeti{TODO: remove the extra $\log\parens *{1/\delta}$ factors. }
\begin{theorem}[Uniform convergence with noiseless measurements]\label{thm:noiseless}

Let $1\le s\le n$, $0<\epsilon\le 1/64$, and $0<\rho<1$. Let
$A\in\R^{m\times n}$ have independent $\cN(0,I_n)$ rows. Given
%\begin{align*}
%m = \tO\parens*{ \frac{s}{\epsilon}\log\parens*{\frac{en}{s}} + \frac{1}{\epsilon}\log\parens*{\frac{1}{\rho}} } 
%\end{align*}
\begin{align} 
\label{eq:suff_measurements}
m \gtrsim \frac{s}{\epsilon}\log\parens*{\frac{en}{s}} \sqrt{\log\parens*{\frac{2e}{\epsilon}}} + \frac{s}{\epsilon}\log^{3/2}\parens*{\frac{2e}{\epsilon}} + \frac{1}{\epsilon}\log\parens*{\frac{1}{\rho}} \sqrt{\log\parens*{\frac{2e}{\epsilon}}} 
\end{align}
noiseless measurements of the form $y = \sign(A\xstar)$, with probability at least $1-\rho$ over the draw of $A$, the following holds: for every $x^\star\in\Sigma_s^n\cap S^{n-1}$, and every initialization $x^{(0)}\in\Sigma_s^n\cap S^{n-1}$,  normalization-free BIHT, when run for $T = 1+\ceil*{\log_2\parens*{\frac{1}{64\epsilon}}}$ iterations, returns an estimate $\widehat{x}_{\mathrm{BIHT}}$ satisfying
\begin{align*}
    \norm{\widehat {x}_{\mathrm{BIHT}}-\xstar}_2 \le \epsilon
\end{align*} 
\end{theorem}

\begin{proof}
Let $c_1, c_2$ be the absolute constants in \cref{fact:gaussian-matrices-satisfy-raic}. Choose an absolute constant $L$, so that
\begin{align*}
2c_1\sqrt{\frac{2}{64L}}+2c_2\frac{1}{64L}
\le
\frac{1}{128}
\end{align*}
and
\begin{align*}
\frac{1}{64}+\frac{c_1}{\sqrt L}+\frac{c_2}{L}
\le
\frac{1-\nicefrac{1}{64}}{16}
\end{align*}
Set $\delta=\epsilon/L$. Since $\epsilon\le 1/64$ and $L$ is an absolute constant, the stated lower bound on $m$, after adjusting the hidden absolute constant, is sufficient to invoke \cref{fact:gaussian-matrices-satisfy-raic} at scale $\delta$. Thus, with probability at least $1-\rho$, the matrix $A$ satisfies $(s,n,\delta,c_1,c_2)$-\RAIC{}.

We condition on this event. The remainder of the proof is deterministic and holds uniformly over all $s$-sparse target vectors, sparse-unit initializations, and hard-thresholding choices. Fix arbitrary $\xstar\in\Sigma_s^n\cap S^{n-1}$ and arbitrary $x^{(0)}\in\Sigma_s^n\cap S^{n-1}$. Let $\xt$ denote the unnormalized BIHT iterates. Whenever $\xt\neq0$, define the current norm $\rt\triangleq\norm{\xt}_2$, the current direction $\ut\triangleq x^{(t)}/\norm{x^{(t)}}_2$, and the current error $e_t\triangleq\norm{\ut-\xstar}_2$.

We first show that the first iterate enters the required constant-sized basin from an arbitrary initialization. Let $S^\star=\supp(\xstar)$, $S_0=\supp(x^{(0)})$, and $S_1=\supp(x^{(1)})$. We will also need the coordinate sets $J_1= S_1\setminus(S^\star\cup S_0)$ and $U_1= S^\star\cup S_0\cup J_1$ in order to invoke \cref{fact:gaussian-matrices-satisfy-raic}. Since $S_1$ has size at most $s$, we have $\card{J_1}\le s$, and by construction $S_1\subseteq U_1$. By the support-restriction property of hard thresholding in \cref{lem:support-restriction}, we have 
\begin{align*}
x^{(1)}
\in
\cT_s\parens*{
\cT_{U_1}\parens*{
x^{(0)}+h_A(\xstar,x^{(0)})
}
}
\end{align*}
Since $x^{(0)}$ is supported on $S_0\subseteq U_1$, this is equivalently
\begin{align*}
x^{(1)}
\in
\cT_s\parens*{
x^{(0)}+h_{A;J_1}(\xstar,x^{(0)})
}
\end{align*}
Applying \RAIC{} to the pair $(\xstar,x^{(0)})$ and the set $J_1$, there exists an error vector $\xi_0$ such that
\begin{align*}
h_{A;J_1}(\xstar,x^{(0)})
=
\xstar-x^{(0)}+\xi_0
\end{align*}
with
\begin{align*}
\norm{\xi_0}_2
\le
c_1\sqrt{\delta\norm{\xstar-x^{(0)}}_2}
+
c_2\delta
\end{align*}
%Therefore, the restricted dense vector is exactly $\xstar+\xi_0$. 
Since $\xstar\in\Sigma_s^n$, \cref{lem:2-approx-projector} gives
\begin{align*}
\norm{x^{(1)}-\xstar}_2
\le 2\norm{x^{(0)}+h_{A;J_1}(x^\star, x^{(0)})-\xstar}_2
\le
2\norm{\xi_0}_2
\end{align*}
Both $\xstar$ and $x^{(0)}$ are unit vectors, which gives the trivial upper bound $\norm{\xstar-x^{(0)}}_2\le2$, and 
\begin{align*}
\norm{x^{(1)}-\xstar}_2
\le
2c_1\sqrt{2\delta}+2c_2\delta
\le
2c_1\sqrt{\frac{2}{64L}}+2c_2\frac{1}{64L}
\le
\frac{1}{128}
\end{align*}
Since $\norm{\xstar}_2=1$, this implies $x^{(1)}\neq0$. Moreover, 
\begin{align*}
\abs*{r^{(1)}-1}
=
\abs*{\norm{x^{(1)}}_2 - \norm{x^\star}_2}
\le
\norm{x^{(1)}-\xstar}_2
\le
\frac{1}{128}
\end{align*}
Finally, since $x^{(1)}=r^{(1)}u^{(1)}$,
\begin{align*}
%e^{(1)}
e_1
=
\norm{u^{(1)}-\xstar}_2
\le
\norm{u^{(1)}-x^{(1)}}_2+\norm{x^{(1)}-\xstar}_2
=
\abs{1-r^{(1)}}+\norm{x^{(1)}-\xstar}_2
\le
\frac{1}{64}
\end{align*}
Next, we prove the one-step contraction that yields the desired induction. 

Fix $t\ge1$, assume $\xt\neq0$, and let $\rt = \norm{\xt}_2, \ut = \xt/\norm{\xt}_2$. Let $S_t=\supp(\xt)$ and $S_{t+1}=\supp(\xtp)$, and define $J_{t+1}=S_{t+1}\setminus(S^\star\cup S_t)$. Then, \cref{lem:support-restriction} together with \cref{lem:sign-norm-invariant} give
\begin{align*}
\xtp
\in
\cT_s\parens*{
\xt+h_{A;J_{t+1}}(\xstar,\ut)
}
\end{align*}
Since $\card{J_{t+1}}\le s$, \RAIC{} applied to the pair $(\xstar,\ut)$ and the set $J_{t+1}$ gives an error vector $\xi_t$ such that
\begin{align*}
h_{A;J_{t+1}}(\xstar,\ut)
=
\xstar-\ut+\xi_t
\end{align*}
and
\begin{align*}
\norm{\xi_t}_2 \le c_1\sqrt{\delta e_t}+c_2\delta
\end{align*}
Substituting $\xt=\rt\ut$ gives the decomposition
\begin{align*}
\xt+h_{A;J_{t+1}}(\xstar,\ut)
&=
\xt + x^\star - \ut + \xi_t\\
&=
\rt\xstar + \rt\ut -\rt x^\star + \xstar -\ut + \xi_t \\
&=
\rt\xstar+(\rt-1)(\ut-\xstar)+\xi_t
\end{align*}
This is where the moving reference enters the proof. The correction estimates $\xstar-\ut$, but it is added to $\rt\ut$. Comparing the result to $\rt\xstar$ instead of $\xstar$ turns the radial discrepancy into the direction-dependent product $(\rt-1)(\ut-\xstar)$.

Since $\xtp \in \cT_s\parens *{\xt+h_{A;J_{t+1}}(\xstar,\ut)}$ and $\rt\xstar\in\Sigma_s^n$, \cref{lem:2-approx-projector} gives
\begin{align}
\norm{\xtp - \rt\xstar}_2&\le 2 \norm{\xt+h_{A;J_{t+1}}(\xstar,\ut) - \rt\xstar}_2\\
&\le 2 \parens *{\abs{\rt-1}e_t+c_1\sqrt{\delta e_t}+c_2\delta}\label{eq:noiseless-direction-one-step}
%\norm{\xtp-\rt\xstar}_2
\end{align}
%Whenever the right-hand side is $\le \rt/2$, \cref{lem:normalization-around-radial-target} gives
%\begin{align*}
%e_{t+1}
%\le
%\frac{8}{\rt} \parens *{\abs{\rt-1}e_t+b_1\sqrt{\delta e_t}+b_2\delta}
%\end{align*}
The same estimate also controls the radial drift, since 
\begin{align*}
\abs{\rtp-\rt} = \abs*{\norm{\xtp}_2-\norm{\rt\xstar}_2}
\le
\norm{\xtp-\rt\xstar}_2
%\le
%2\parens *{\abs{\rt-1}e_t+b_1\sqrt{\delta e_t}+b_2\delta}
\end{align*}
where the last inequality is due to the triangle inequality. Substituting the upper bound from \cref{eq:noiseless-direction-one-step} gives
\begin{align}\label{eq:noiseless-radial-one-step}
    \abs{\rtp-\rt} \le 2\parens *{\abs{\rt-1}e_t+c_1\sqrt{\delta e_t}+c_2\delta}
\end{align}

It remains to close the induction. Let
\begin{align*}
T
=
1+\ceil*{\log_2\parens*{\frac{1}{64\epsilon}}}
\end{align*}
as stated. We claim that for every $1\le t\le T$,
\begin{align*}
\xt\neq0,
\qquad
e_t\le\max\set{2^{-(t-1)}/64,\epsilon},
\qquad
\abs*{\rt-1}\le\frac{1}{64}
\end{align*}
The base case was proved above. Assume the claim holds at some time $t<T$. Since $\epsilon\le1/64$ and $e_t\le1/64$, we have $\max\set{e_t,\epsilon}\le1/64$. Moreover, $\abs{\rt-1}\le1/64$. If $e_t\ge\epsilon$, then $\delta=\epsilon/L\le e_t/L$, and
\begin{align*}
\abs{\rt-1}e_t+c_1\sqrt{\delta e_t}+c_2\delta
\le
\parens*{
\frac{1}{64}+\frac{c_1}{\sqrt L}+\frac{c_2}{L}
}e_t
\end{align*}
If instead $e_t\le\epsilon$, then similarly
\begin{align*}
\abs{\rt-1}e_t+c_1\sqrt{\delta e_t}+c_2\delta
\le
\parens*{
\frac{1}{64}+\frac{c_1}{\sqrt L}+\frac{c_2}{L}
}\epsilon
\end{align*}
Together, we have
\begin{align*}
\abs{\rt-1}e_t+c_1\sqrt{\delta e_t}+c_2\delta
\le
C\max\set{e_t,\epsilon}
\end{align*}
where $C =\nicefrac{1}{64}+c_1/\sqrt L+c_2/L$ is an absolute constant. Moreover, by our choice of $L$, $C\le(1-\nicefrac{1}{64})/16$.

Since $\rt\ge1-\nicefrac{1}{64}$ and $\max\set{e_t,\epsilon}\le1/64$,
\begin{align*}
2\parens *{\abs{\rt-1}e_t+c_1\sqrt{\delta e_t}+c_2\delta}
\le
2C\cdot\frac{1}{64}
\le
\frac{1-\nicefrac{1}{64}}{8}\cdot\frac{1}{64}
\le
\frac{1-\nicefrac{1}{64}}{2}
\le
\frac{\rt}{2}
\end{align*}
Therefore, the normalization estimate of \cref{lem:normalization-around-radial-target} applies, and
\begin{align*}
e_{t+1}
\le
\frac{8}{\rt}\parens *{\abs{\rt-1}e_t+c_1\sqrt{\delta e_t}+c_2\delta}
\le
\frac{8}{1-\nicefrac{1}{64}}C\max\set{e_t,\epsilon}
\le
\frac{1}{2}\max\set{e_t,\epsilon}
\end{align*}
From the induction hypothesis, we have $e_t\le\max\set{2^{-(t-1)}/64,\epsilon}$, which then gives the upper bound
\begin{align*}
e_{t+1}
\le
\frac{1}{2}\max\set{2^{-(t-1)}/64,\epsilon}
\le
\max\set{2^{-t}/64,\epsilon}
\end{align*}
This proves the directional part of the induction. For the radial drift, we simply substitute the derived upper bounds into the one-step estimate of \cref{eq:noiseless-radial-one-step}:
\begin{align*}
\abs{\rtp-\rt}
&\le
2\parens *{\abs{\rt-1}e_t+c_1\sqrt{\delta e_t}+c_2\delta}\\
&\le
2C\max\set{e_t,\epsilon}\\
&\le
\frac{1-\nicefrac{1}{64}}{8}\max\set{e_t,\epsilon}
\end{align*}
Summing over the radial increments from time $1$ to $t$ gives
\begin{align*}
\abs{\rtp-1}
&\le
\abs{r^{(1)}-1}
+
\sum_{j=1}^t\abs{r^{(j+1)}-r^{(j)}} \\
&\le
\frac{1}{128}
+
\frac{1-\nicefrac{1}{64}}{8}
\sum_{j=1}^t
\max\set{2^{-(j-1)}/64,\epsilon}
\end{align*}
Since $t<T$, it suffices to bound the sum up to $T-1$. We have
\begin{align*}
\sum_{j=1}^{T-1}\max\set{2^{-(j-1)}/64,\epsilon}
\le
\sum_{j=1}^{T-1}\frac{2^{-(j-1)}}{64}
+
(T-1)\epsilon
\end{align*}
The geometric sum is at most $1/32$. For the second term, set $q=1/(64\epsilon)\ge1$. Then $T-1=\ceil*{\log_2 q}\le1+\log_2 q$, and $1+\log_2 q\le2q$ for all $q\ge1$. Therefore, $(T-1)\epsilon\le1/32$, and 
\begin{align*}
\sum_{j=1}^{T-1}\max\set{2^{-(j-1)}/64,\epsilon}
\le
\frac{1}{16}
\end{align*}
Substituting this into the norm estimate yields
\begin{align*}
\abs{\rtp-1}
\le
\frac{1}{128}
+
\frac{1-\nicefrac{1}{64}}{128}
<
\frac{1}{64}
\end{align*}
This closes the induction. In particular, $x^{(1)},\ldots,x^{(T)}$ are all nonzero.

At the final time $T$, the directional envelope gives
\begin{align*}
e_T
\le
\max\set{2^{-(T-1)}/64,\epsilon}
\end{align*}
By definition of $T$, $2^{-(T-1)}/64\le\epsilon$, so $e_T\le\epsilon$. Since \cref{alg:biht} returns the final normalized estimate $\widehat{x}_{\mathrm{BIHT}}=x^{(T)}/\norm{x^{(T)}}_2$, this is exactly
\begin{align*}
\norm{\widehat{x}_{\mathrm{BIHT}}-\xstar}_2
\le
\epsilon
\end{align*}
and this completes the proof. 
\end{proof}

\paragraph{Remark.} The guarantee of the theorem persists for $\floor*{\nicefrac{1}{64\epsilon^2}}$ additional iterations. Indeed, the theorem proves $e_T\le\eps$ and $\abs{r^{(T)}-1}\le1/64$. We can simply continue the induction with $e_t\le\eps$ and $\abs{\rt-1}\le1/16$. RAIC bounds the norm of the restricted correction as $\norm{h_{A;J_{t+1}}(\xstar,\ut)}_2 \le 17\eps/16$, while noiselessness implies $\langle h_t,u^{(t)}\rangle\le0\le\langle h_t,x^\star\rangle$ via \cref{eq:hA-def}. By \cref{lem:support-restriction,lem:sign-norm-invariant} and optimality of hard thresholding,
\begin{align*}
    \rt\left(1-\frac{\eps^2}{2}\right)
    \le \rtp
    \le \sqrt{(\rt)^2+\left(\frac{17\eps}{16}\right)^2}
\end{align*}
Since $15/16\le\rt\le17/16$, these inequalities imply $\abs{\rtp-\rt}\le\eps^2$. Meanwhile, \cref{eq:noiseless-direction-one-step} and \cref{lem:normalization-around-radial-target} give $e_{t+1}\le111\eps/120<\eps$. Summing over the increments from $T$ gives
\begin{align*}
    \abs{r^{(t)}-1}
    \le \frac1{64}+(t-T)\eps^2
    \le \frac1{32}<\frac1{16}
\end{align*}
throughout the stated window, and this closes the induction.

\subsubsection{Late iterate behavior: conditional all-time theorem}

The window above is finite and computable, given $\epsilon$. We now turn to the behavior of the iterates as $t\to\infty$. With more measurements, the iterates of normalization-free BIHT remain $\eps$-accurate at every $t\ge T$ (\cref{thm:noiseless-all-time}), and at the optimal sample complexity of \cref{thm:noiseless}, this holds for every target whose nonzero entries are at least $\eps$ in magnitude. We defer the proofs of this section to \cref{app:defer_late_iterate_noiseless}.

\begin{theorem}[All-time guarantee]
\label{thm:noiseless-all-time}
Fix $0<\mu\le1$. Assume the hypotheses of \cref{thm:noiseless}, with $\eps$ replaced by $\min\{\eps,\mu\}$ in Eq.~\eqref{eq:suff_measurements}, and let $T=1+\ceil{\log_2(1/(64\eps))}$. Then, with probability at least $1-\rho$ over the draw of $A$, the following holds simultaneously for every target $\xstar\in\Sigma_s^n\cap S^{n-1}$ whose nonzero entries have magnitude at least $\mu$ and every sparse-unit initialization: the normalization-free BIHT iterates satisfy $\xt\neq0$ and
\begin{align*}
    \norm*{\frac{\xt}{\norm{\xt}_2}-\xstar}_2\le\eps
    \qquad\text{for every integer }t\ge T
\end{align*}
In particular, when $\mu\ge\eps$, the measurement requirement is the same as in \cref{thm:noiseless}.
\end{theorem}

\begin{proof}
    See \cref{app:defer_late_iterate_noiseless}.
\end{proof}

This dependence of the number of measurements on $\mu$ cannot be removed by an argument that uses only RAIC. The next result constructs a sensing matrix that satisfies the same RAIC as in the proof of \cref{thm:noiseless}, a target with a nonzero entry smaller than $\eps$, and an initialization, for which the directional error of normalization-free BIHT stays below $\eps$ until some iteration beyond the stable window, at which it exceeds $1/2$.

\begin{theorem}[Late failure under RAIC]
\label{thm:noiseless-late-failure}
Let $0<\eps\le1/64$, and let $L,c_1,c_2$ be as in the proof of \cref{thm:noiseless}. Set $\delta=\eps/L$ and $T=1+\ceil{\log_2(1/(64\eps))}$.

For every sufficiently large $m$, there exist $A\in\R^{m\times3}$ satisfying $(2,3,\delta,c_1,c_2)$-RAIC and $\xstar,x^{(0)}\in\Sigma_2^3\cap S^2$ with
$\min_{i\in\supp(\xstar)}\abs{x_i^\star}<\eps$ such that the following holds. With sparsity $s=2$ and noiseless measurements $y=\sign(A\xstar)$, normalization-free BIHT has a finite time $\tau>T+\floor*{1/(64\eps^2)}$ for which all iterates through $\tau$ are nonzero and 
\begin{align*}
    e_t<\eps\quad(0\le t<\tau),
    \qquad
    e_\tau>\frac12
\end{align*}
\end{theorem}

\begin{proof}
    See \cref{app:defer_late_iterate_noiseless}.
\end{proof}

\subsection{Adversarial Corruptions}

In the presence of sign flips, a general last-iterate convergence theorem for \BIHT{}, as proposed in \textcite{jacques2013robust}, with constant step-size,  is impossible. Unlike in the noiseless regime, no amount of additional measurements can prevent this phenomenon. 
\subsubsection{Late iterate behavior: no all-time guarantee is feasible}
\begin{theorem}[Iterates flip signs infinitely often]\label{thm:adversarial-lower-bound}
Let $n=s=1$, $x^\star=1$, and $A\in\R^{m\times1}$ have entries $a_1,\ldots,a_m$ with $a_i\neq0$. Let $\emptyset\neq D\subsetneq[m]$ be a nonempty proper subset of indices, and define corrupted labels by
\begin{align*}
y_i=\begin{cases}
-\sign(a_i),& i\in D,\\
\sign(a_i),& i\notin D.
\end{cases}
\end{align*}
Set
\begin{align*}
\alpha\triangleq\frac{\eta}{m}\sum_{i\in D}\abs{a_i},
\qquad
\gamma\triangleq\frac{\eta}{m}\sum_{i\notin D}\abs{a_i}.
\end{align*}
Then, $\alpha>0$ and $\gamma>0$. Let $x^{(0)}\in\R\setminus\set{0}$ be
arbitrary. If the normalization-free \BIHT{} iterates never hit zero, then for
every $t\ge0$,
\begin{align*}
x^{(t+1)}=\begin{cases}
x^{(t)}-\alpha,& \text{ if } x^{(t)}>0,\\
x^{(t)}+\gamma,& \text{ if }x^{(t)}<0.
\end{cases}
\end{align*}
Consequently, $x^{(t)}$ changes sign infinitely often, and
\begin{align*}
\limsup_{t\to\infty}
\norm*{
\frac{x^{(t)}}{\abs{x^{(t)}}}
-
x^\star
}_2
=
2.
\end{align*}
If the entries of $A$ are independent standard Gaussians, $D$ is fixed, and $x^{(0)}\neq0$ is fixed independently of $A$, then $\P *{\exists t \ge 0 : \xt = 0 } = 0$. 
\end{theorem}

\begin{proof}
Since $n=s=1$, hard thresholding is the identity. Fix any time $t$ such that $x^{(t)}\neq0$. If $x^{(t)}>0$, then $\sign(a_ix^{(t)})=\sign(a_i)$ for every $i$, and
\begin{align*}
\frac{\eta}{2m}\sum_{i=1}^m a_i\parens*{y_i-\sign(a_ix^{(t)})}
&=
\frac{\eta}{2m}\sum_{i\in D}
a_i\parens*{-\sign(a_i)-\sign(a_i)}\\
&=
-\frac{\eta}{m}\sum_{i\in D}\abs{a_i}\\
&=
-\alpha.
\end{align*}
Thus $x^{(t+1)}=x^{(t)}-\alpha$ whenever $x^{(t)}>0$.

If $x^{(t)}<0$, then $\sign(a_ix^{(t)})=-\sign(a_i)$ for every $i$, and
\begin{align*}
\frac{\eta}{2m}\sum_{i=1}^m a_i\parens*{y_i-\sign(a_ix^{(t)})}
&=
\frac{\eta}{2m}\sum_{i\notin D}
a_i\parens*{\sign(a_i)+\sign(a_i)}\\
&=
\frac{\eta}{m}\sum_{i\notin D}\abs{a_i}\\
&=
\gamma.
\end{align*}
Thus $x^{(t+1)}=x^{(t)}+\gamma$ whenever $x^{(t)}<0$.

Because $\alpha,\gamma>0$, a positive iterate is decreased by $\alpha$ at every positive step until it becomes negative, unless it hits zero exactly. Similarly, a negative iterate is increased by $\gamma$ at every negative step until it becomes positive, unless it hits zero exactly. Under the hypothesis that the iterates never hit zero, the sign therefore changes infinitely often.

Whenever $x^{(t)}<0$, the normalized output is
$x^{(t)}/\abs{x^{(t)}}=-1$. Since $x^\star=1$, we have
\begin{align*}
\norm*{
\frac{x^{(t)}}{\abs{x^{(t)}}}
-
x^\star
}_2
=
2
\end{align*}
for infinitely many $t$, proving the limsup claim. It remains to verify the probability-zero claim. Suppose, for contradiction, $\P *{\exists t \ge 0 : \xt = 0 } > 0$. Then, there exists a fixed $t\ge1$ such that
\begin{align*}
	\P{x^{(0)},\ldots,x^{(t-1)}\neq0,\ x^{(t)}=0}>0
\end{align*}
On this event, the recursion holds for the first $t$ updates, and there must exist nonnegative integers $p,q$ with $p+q=t$ such that
\begin{align*}
	x^{(0)}-p\alpha+q\gamma=0
\end{align*}
Since there are only finitely many such pairs $(p,q)$, positive probability of the preceding event implies that, for a fixed pair $(p,q)$ with $p+q=t$,
\begin{align*}
	\P*{x^{(0)}-p\alpha+q\gamma=0}>0
\end{align*}
But $p$ and $q$ cannot both be zero, since $p+q=t\ge1$. Moreover, since $D$ and $\bracks *{m} \setminus D$ are both non-empty, $\alpha$ and $\gamma$ are independent non-atomic random variables, and so is the nontrivial affine combination $x^{(0)}-p\alpha+q\gamma$. Therefore,  $\P*{x^{(0)}-p\alpha+q\gamma=0}=0$ and we arrive at a contradiction.
\end{proof}

\paragraph{Per-iterate normalization escapes this instance}\label{rem:normalization-escapes-lb-instance}
Consider the same one-dimensional instance and let $\ut\in\set*{\pm 1}$ denote the normalized BIHT iterate. Since $n=s=1$, thresholding is the identity. Before normalization, the next iterate is of the form
\begin{align*}
    \tilde u^{(t+1)} = \ut +\frac{\eta}{2m}\sum_{i=1}^m a_i\parens*{y_i-\sign(a_i u^{(t)})}
\end{align*}
With $\alpha, \gamma$ as defined in \cref{thm:adversarial-lower-bound}, 
\begin{align*}
\tilde u^{(t+1)} 
=
\begin{cases}
1-\alpha, & \text{ if }u^{(t)}=1,\\
-1+\gamma, & \text{ if }u^{(t)}=-1.
\end{cases}
\end{align*}
Therefore, whenever the pre-normalized scalar is non-zero, 
\begin{align*}
\utp
= \frac{\tilde{u}^{(t+1)}}{\abs*{\tilde{u}^{(t+1)}}} = 
\begin{cases}
\frac{1-\alpha}{\abs*{1-\alpha}}, & \text{ if }u^{(t)}=1,\\
\frac{-1+\gamma}{\abs*{-1+\gamma}}, & \text{ if }u^{(t)}=-1.
\end{cases}
\end{align*}
On the event $\alpha<1<\gamma$, both possible current signs are mapped to the correct sign after normalization. In fact, this event is the finite sample analog of having fewer flipped signs than clean signs. Indeed, since $\eta\E{\abs{a_i}}=\sqrt{2\pi}\cdot\sqrt{2/\pi}=2$, for fixed $D$ we have
\begin{align*}
\E{\alpha}=2\frac{\abs{D}}{m},
\qquad
\E{\gamma}=2\parens*{1-\frac{\abs{D}}{m}}.
\end{align*}
Moreover, $\abs{a_i}-\E{\abs{a_i}}$ is sub-Gaussian. If
$\abs{D}\le \tau m$ for a fixed constant $\tau<1/2$, then sub-Gaussian concentration gives
\begin{align*}
\P{\alpha<1<\gamma}
\ge
1-2\exp\parens*{-c(1-2\tau)^2m}
\end{align*}
for an absolute constant $c>0$. 
%\end{remark}

While both algorithms receive identical pushes $-\alpha$ and $+\gamma$ to force a sign switch, normalization resets the iterate to unit magnitude after every update. From the correct sign, the corrupted measurements can overturn the normalized iterate only if $\alpha>1$. In the regime discussed above, this failure event has exponentially small probability for Gaussian sensing matrices. Without normalization, there is no such reset -- a positive unnormalized iterate repeatedly loses $\alpha$, which eventually carries it across the origin over a long enough time frame. If only one measurement is flipped, $\alpha$ is of order $1/m$. An order-one positive unnormalized iterate crosses the origin after order $m$ iterations, while the normalized iterate is correct from the first iterate onward (on the event $\alpha<1<\gamma$).

\subsubsection{Early iterate behavior before oscillations}
The scalar lower bound rules out uniform last-iterate convergence under corruptions, but it does not rule out accurate early iterates. The next result gives the corresponding positive statement: the unnormalized trajectory reaches the same robust error floor as normalized \BIHT{} at an early, certifiable iterate.
\begin{theorem}[Early accurate iterate under adversarial sign flips]
\label{thm:adversarial-positive}
Let $1\le s\le n$, $0<\epsilon\le 1/64$, $0<\rho<1$, and $0\le \tau\le 1$. Let
$A\in\R^{m\times n}$ have independent $\cN(0,I_n)$ rows. Suppose
\begin{align*}
	  m \gtrsim \frac{s}{\epsilon}\log\parens*{\frac{en}{s}}
         \sqrt{\log\parens*{\frac{2e}{\epsilon}}}
       + \frac{s}{\epsilon}\log^{3/2}\parens*{\frac{2e}{\epsilon}}
       + \frac{1}{\epsilon}\log\parens*{\frac{1}{\rho}}
         \sqrt{\log\parens*{\frac{2e}{\epsilon}}}
\end{align*}
Then, with probability at least $1-\rho$ over the draw of $A$, the following holds simultaneously for every $\xstar\in\Sigma_s^n\cap S^{n-1}$, every corrupted observation vector $y\in\set{\pm1}^m$ satisfying $d_{\mathrm H}(y,\sign(A\xstar))\le \tau m$, and every initialization $x^{(0)}\in\Sigma_s^n\cap S^{n-1}$: Consider the normalization-free BIHT iterates run with measurements $y$. There exist absolute constants $c_1, c_2>0$ such that, with the convention $\tau\sqrt{\log(2e/\tau)}=0$ when $\tau=0$, if
\begin{align*}
	\gamma(\epsilon,\tau)
	\triangleq
    \epsilon 
    + 
    c_1\sqrt{\epsilon\tau}
    +
    c_2 \tau\sqrt{\log\parens*{\frac{2e}{\tau}}}
	%C\parens*{
	%	\epsilon
	%	+
	%	\sqrt{\epsilon\tau}
	%	+
	%	\tau\sqrt{\log\parens*{\frac{2e}{\tau}}}
	%}
	\le
	\frac{1}{64},
\end{align*}
then the first hitting time for normalization-free BIHT
\begin{align*}
	t_{\mathrm{hit}}
	\triangleq
	\inf\set*{
		t\ge1
		\given
		x^{(t)}\neq0
		\text{ and }
		\norm*{
			\frac{x^{(t)}}{\norm{x^{(t)}}_2}
			-
			\xstar
		}_2
		\le
		\gamma(\epsilon,\tau)
	}
\end{align*}
satisfies
\begin{align*}
	t_{\mathrm{hit}}
	\le
	T
	\triangleq
	1+\ceil*{\log_2\parens*{\frac{1}{64\gamma(\epsilon,\tau)}}}.
\end{align*}
Moreover, for every $1\le t\le t_{\mathrm{hit}}$, the iterate $x^{(t)}$ is nonzero and satisfies $\abs*{\norm{x^{(t)}}_2-1}\le\frac{1}{64}$.  %Equivalently,
%\begin{align*}
%	\min_{1\le t\le T}
%	\norm*{
%		\frac{x^{(t)}}{\norm{x^{(t)}}_2}
%		-
%		\xstar
%	}_2
%	\le
%	\gamma(\epsilon,\tau).
%\end{align*}
\end{theorem}
\begin{proof}
%Let $b_1 = \sqrt{\frac{3\pi}{380}}\parens*{1+\frac{16\sqrt2}{3}}, b_2 = \frac{3}{380}\parens*{1+\frac{4\pi}{3}+\frac{8\sqrt{3\pi}}{3}+8\sqrt{6\pi}}, b_3 = \frac{(12+\sqrt3)\sqrt{\pi}}{\sqrt 380}, b_4 = 2+4\sqrt{\pi}$ be the universal constants in \cref{fact:gaussian-matrices-satisfy-adv-raic}. 
Let $b_1, b_2, b_3, b_4$ be the universal constants in \cref{fact:gaussian-matrices-satisfy-adv-raic}. Choose absolute constants $c_2 = c_1^2$ and $L$ such that

\begin{align}
2b_1\sqrt{\frac{2}{64c_2L}}
+
2b_2\frac{1}{64c_2L}
+
2b_3\frac{1}{64c_2\sqrt L}
+
2b_4\frac{1}{64c_2}
&\le
\frac{1}{128}
\label{eq:adv-base-constant-requirement}
\end{align}
and
\begin{align}
\frac{1}{64}
+
\frac{b_1}{\sqrt{c_2L}}
+
\frac{b_2}{c_2L}
+
\frac{b_3}{c_2\sqrt L}
+
\frac{b_4}{c_2}
&\le
\frac{1-\nicefrac{1}{64}}{16}.
\label{eq:adv-contraction-constant-requirement}
\end{align}

Set $\delta=\nicefrac{\epsilon}{c_2L}$. Since $L$ is an absolute constant, the stated lower bound on $m$, after adjusting the hidden absolute constant, is sufficient to invoke \cref{fact:gaussian-matrices-satisfy-adv-raic} at scale $\delta$. We condition on the resulting high-probability event. The remainder of the proof is deterministic and holds uniformly over the target vector, the corrupted measurements, the sparse-unit initialization, and the hard-thresholding choices.

Fix arbitrary $\xstar\in\Sigma_s^n\cap S^{n-1}$, arbitrary $y\in\set{\pm1}^m$ satisfying $d_{\mathrm H}(y,\sign(A\xstar))\le\tau m$, arbitrary $x^{(0)}\in\Sigma_s^n\cap S^{n-1}$, and arbitrary hard-thresholding choices. Define $f:\R^n\to\set{\pm1}^m$ by
\begin{align*}
f(z)
=
\begin{cases}
y, & z=\xstar,\\
\sign(Az), & z\neq\xstar.
\end{cases}
\end{align*}
Then $f\in\cF_A(\tau)$, and the \BIHT{} update with measurements $y$ can be written as
\begin{align*}
x^{(t+1)}
\in
\cT_s\parens*{
x^{(t)}
+
h_{f;A}(\xstar,x^{(t)})
}.
\end{align*}

Whenever $\xt\neq0$, define $\rt\triangleq\norm{\xt}_2$, $\ut\triangleq\xt/\norm{\xt}_2$, and $e_t\triangleq\norm{\ut-\xstar}_2$.

We first prove that the first iterate enters the same constant-sized basin as in the noiseless proof. Let $S^\star\triangleq\supp(\xstar)$, $S_0\triangleq\supp(x^{(0)})$, and $S_1\triangleq\supp(x^{(1)})$. Define $J_1\triangleq S_1\setminus(S^\star\cup S_0)$ and $U_1\triangleq S^\star\cup S_0\cup J_1$. Since $S_1$ has size at most $s$, $\card{J_1}\le s$, and by construction $S_1\subseteq U_1$. By \cref{lem:support-restriction},
\begin{align*}
x^{(1)}
\in
\cT_s\parens*{
\cT_{U_1}
\parens*{
x^{(0)}+h_{f;A}(\xstar,x^{(0)})
}
}.
\end{align*}
Since $x^{(0)}$ is supported on $S_0\subseteq U_1$, this is equivalently
\begin{align*}
x^{(1)}
\in
\cT_s\parens*{
x^{(0)}+h_{f;A;J_1}(\xstar,x^{(0)})
}.
\end{align*}
From \cref{fact:gaussian-matrices-satisfy-adv-raic} applied to the pair $(\xstar,x^{(0)})$, the function $f$, and the set $J_1$, there exists an error vector $\xi_0$ such that
\begin{align*}
h_{f;A;J_1}(\xstar,x^{(0)})
=
\xstar-x^{(0)}+\xi_0
\end{align*}
and
\begin{align*}
\norm *{\xi_0}_2
\le
b_1\sqrt{\delta\norm{\xstar-x^{(0)}}_2}
+
b_2\delta
+
b_3\sqrt{\delta\tau}
+
b_4\tau\sqrt {\log\parens *{\frac{2e}{\tau}}}.
\end{align*}
The restricted dense vector is $\xstar+\xi_0$. Since $\xstar\in\Sigma_s^n$, \cref{lem:2-approx-projector} gives
\begin{align*}
\norm{x^{(1)}-\xstar}_2
\le
2\norm{\xi_0}_2.
\end{align*}
Since both $\xstar$ and $x^{(0)}$ are unit vectors, $\norm{\xstar-x^{(0)}}_2\le 2$. Moreover, by our choice of constants, $c_1 ^2 = c_2$, and 
\begin{align*}
\gamma(\epsilon, \tau) = c_2\parens*{
\frac{\epsilon}{c_2}
+
\sqrt{\frac{\epsilon}{c_2}\tau}
+
\tau\sqrt{\log \parens*{\frac{2e}{\tau}}}
} \le \frac{1}{64}
\end{align*}
Therefore, each non-negative term inside the parenthesis on the left-hand side is at most $1/(64c_2)$. Since $\delta = \nicefrac{\epsilon}{c_2L}$, we have 
%\begin{align*}
%    \epsilon \le \epsilon + \sqrt{\epsilon\tau} + \tau\sqrt{\log \parens *{\frac{2e}{\tau}}} \le \frac{1}{64 C}
%\end{align*}
%Therefore, with $\delta = \epsilon / L$, we have
\begin{align*}
\norm{x^{(1)}-\xstar}_2
&\le
2b_1\sqrt{2\delta}
+
2b_2\delta
+
2b_3\sqrt{\delta\tau}
+
2b_4\tau\sqrt{\log\parens *{\frac{2e}{\tau}}} \\
&\le
2b_1\sqrt{\frac{2}{64c_2L}}
+
2b_2\frac{1}{64c_2L}
+
2b_3\frac{1}{64c_2\sqrt L}
+
2b_4\frac{1}{64c_2}
\le
\frac{1}{128},
\end{align*}
where the last inequality is due to \cref{eq:adv-base-constant-requirement}. Therefore, $x^{(1)}\neq0$, and via triangle inequality, 
\begin{align*}
\abs*{r^{(1)}-1} = \abs *{\norm*{x^{(1)}}_2 - \norm*{x^\star}_2}
\le
\norm{x^{(1)}-\xstar}_2
\le
\frac{1}{128}.
\end{align*}
Since $x^{(1)}=r^{(1)}u^{(1)}$, we also have
\begin{align*}
e_1
=
\norm{u^{(1)}-\xstar}_2
\le
\abs*{1-r^{(1)}}+\norm{x^{(1)}-\xstar}_2
\le
\frac{1}{64}.
\end{align*}
If $e_1\le\gamma(\epsilon,\tau)$, then $t_{\mathrm{hit}}=1$, and the theorem follows. We therefore assume below that $e_1>\gamma(\epsilon,\tau)$.
Next, we will prove the one-step estimate that holds until hitting time. Fix $t\ge1$ and assume $\xt\neq0$. Let $S_t\triangleq\supp(\xt)$ and $S_{t+1}\triangleq\supp(\xtp)$, and define $J_{t+1}\triangleq S_{t+1}\setminus(S^\star\cup S_t)$. Since positive rescaling of the second argument does not change the sign correction, \cref{lem:support-restriction}, together with \cref{lem:sign-norm-invariant}, gives
\begin{align*}
\xtp
\in
\cT_s\parens*{
\xt+h_{f;A;J_{t+1}}(\xstar,\ut)
}.
\end{align*}
Applying adversarial \RAIC{} to the pair $(\xstar,\ut)$, the function $f$, and the set $J_{t+1}$, there exists an error vector $\xi_t$ such that
\begin{align*}
h_{f;A;J_{t+1}}(\xstar,\ut)
=
\xstar-\ut+\xi_t
\end{align*}
and
\begin{align*}
\norm{\xi_t}_2
\le
b_1\sqrt{\delta e_t}
+
b_2\delta
+
b_3\sqrt{\delta\tau}
+
b_4\tau\sqrt{\log \parens *{\frac{2e}{\tau}}}.
\end{align*}
Substituting $\xt=\rt\ut$ gives the same moving-reference decomposition as in the noiseless proof:
\begin{align*}
\xt+h_{f;A;J_{t+1}}(\xstar,\ut)
=
\rt\xstar+(\rt-1)(\ut-\xstar)+\xi_t.
\end{align*}
Since $\xtp \in \cT_s\parens *{\xt + h_{f;A;J_{t+1}}(\xstar,\ut)}$ and $\rt\xstar \in \Sigma_s^n$, \cref{lem:2-approx-projector} gives
\begin{align*}
\norm{\xtp-\rt\xstar}_2
&\le 2 \norm *{\xt + h_{f;A;J_{t+1}}(x^\star, \ut)- \rt x^\star}_2\\
&\le
2 \parens *{\abs*{\rt-1}e_t
+
b_1\sqrt{\delta e_t}
+
b_2\delta
+
b_3\sqrt{\delta\tau}
+
b_4\tau\sqrt{\log \parens*{\frac{2e}{\tau}}}}.
\end{align*}
If $t<t_{\mathrm{hit}}$, then $e_t>\gamma(\epsilon,\tau)$. Since
\begin{align*}
\gamma(\epsilon,\tau)
=
c_2\parens*{\frac{\epsilon}{c_2}+\sqrt{\frac{\epsilon\tau}{c_2}}+\tau\sqrt{\log \parens *{\frac{2e}{\tau}}}}
\end{align*}
and all three terms inside the parentheses are non-negative, we have
\begin{align*}
\epsilon< e_t,
\qquad
\sqrt{\frac{\epsilon\tau}{c_2}}<\frac{e_t}{c_2},
\qquad
\tau\sqrt{\log\parens *{\frac{2e}{\tau}}}<\frac{e_t}{c_2}.
\end{align*}
Using $\delta=\epsilon/c_2L$, we get
\begin{align*}
\sqrt{\delta e_t}
\le
\frac{e_t}{\sqrt{c_2L}},
\qquad
\delta
\le
\frac{e_t}{c_2L},
\qquad
\sqrt{\delta\tau}
\le
\frac{e_t}{c_2\sqrt L}.
\end{align*}
Therefore, whenever $t<t_{\mathrm{hit}}$ and $\abs*{r^{(t)}-1}\le1/64$,

\begin{align}\label{eq:dt-le-bt}
\abs*{r^{(t)}-1}e_t
+
b_1\sqrt{\delta e_t}
+
b_2\delta
+
b_3\sqrt{\delta\tau}
+
b_4\tau\sqrt{\log\parens *{\frac{2e}{\tau}}} 
\le
\parens *{\frac{1}{64}
+
\frac{b_1}{\sqrt{c_2L}}
+
\frac{b_2}{c_2L}
+
\frac{b_3}{c_2\sqrt L}
+
\frac{b_4}{c_2}} e_t.
\end{align}
By \cref{eq:adv-contraction-constant-requirement}, the right-hand side is bounded above by $\parens *{1-\nicefrac{1}{64}}e_t/16 =63 e_t/1024 $.

We claim that, for every $1\le t\le\min\set{T,t_{\mathrm{hit}}}$,
\begin{align*}
x^{(t)}\neq0,
\qquad
e_t\le\frac{2^{-(t-1)}}{64},
\qquad
\abs*{r^{(t)}-1}\le\frac{1}{64}.
\end{align*}
The base case was proved above. Assume the claim holds at time $t<\min\set{T,t_{\mathrm{hit}}}$. Then $t<t_{\mathrm{hit}}$, so \cref{eq:dt-le-bt} holds. Since $e_t\le1/64$ and $\rt\ge63/64$,
\begin{align*}
2\parens *{\abs*{r^{(t)}-1}e_t
+
b_1\sqrt{\delta e_t}
+
b_2\delta
+
b_3\sqrt{\delta\tau}
+
b_4\tau\sqrt{\log\parens *{\frac{2e}{\tau}}} }
&\le
2\parens*{\frac{1}{64}
+
\frac{b_1}{\sqrt{c_2L}}
+
\frac{b_2}{c_2L}
+
\frac{b_3}{c_2\sqrt L}
+
\frac{b_4}{c_2}} \, e_t\\
&\le
2\cdot \frac{63}{1024}\cdot \frac{1}{64}\\
&<
\frac{63}{128}
\le
\frac{\rt}{2}.
\end{align*}

Then, \cref{lem:normalization-around-radial-target} applies around the radial target $\rt\xstar$, and gives

\begin{align*}
e_{t+1}
&\le
\frac{8}{r^{(t)}} \parens*{ \abs*{r^{(t)}-1}e_t
+
b_1\sqrt{\delta e_t}
+
b_2\delta
+
b_3\sqrt{\delta\tau}
+
b_4\tau\sqrt{\log\parens *{\frac{2e}{\tau}}}}\\
&\le
\frac{8}{63/64}\parens *{\frac{1}{64}
+
\frac{b_1}{\sqrt{c_2L}}
+
\frac{b_2}{c_2L}
+
\frac{b_3}{c_2\sqrt L}
+
\frac{b_4}{c_2}} e_t\\
&\le
\frac{512}{63}\parens *{\frac{1-\nicefrac{1}{64}}{16}} e_t\\
&\le
\frac{1}{2}e_t.
\end{align*}

The same distance estimate also controls the radial drift. 
\begin{align*}
\abs*{r^{(t+1)}-r^{(t)}}
&=
\abs*{\norm{x^{(t+1)}}_2-\norm{r^{(t)}\xstar}_2}\\
&\le
\norm{x^{(t+1)}-r^{(t)}\xstar}_2\\
&\le
2\parens *{\abs*{r^{(t)}-1}e_t
+
b_1\sqrt{\delta e_t}
+
b_2\delta
+
b_3\sqrt{\delta\tau}
+
b_4\tau\sqrt{\log\parens *{\frac{2e}{\tau}}} }\\
&\le
2\parens*{\frac{1}{64}
+
\frac{b_1}{\sqrt{c_2L}}
+
\frac{b_2}{c_2L}
+
\frac{b_3}{c_2\sqrt L}
+
\frac{b_4}{c_2}}\, e_t\\
&\le
\frac{63}{512}e_t.
\end{align*}
Summing the radial increments gives
\begin{align*}
\abs*{r^{(t+1)}-1}
&\le
\abs*{r^{(1)}-1}
+
\sum_{j=1}^{t}
\abs*{r^{(j+1)}-r^{(j)}}\\
&\le
\frac{1}{128}
+
\frac{63}{512}
\sum_{j=1}^{t}e_j\\
&\le
\frac{1}{128}
+
\frac{63}{512}
\sum_{j=1}^{\infty}\frac{2^{-(j-1)}}{64}\\
&=
\frac{1}{128}
+
\frac{63}{512}\cdot\frac{1}{32}
<
\frac{1}{64}.
\end{align*}
This closes the induction. It remains to bound the hitting time. Suppose, for contradiction, that $t_{\mathrm{hit}}>T$. Then the induction applies through time $T$, and
\begin{align*}
e_T
\le
\frac{2^{-(T-1)}}{64}.
\end{align*}
By definition of $T$,
\begin{align*}
\frac{2^{-(T-1)}}{64}
\le
\gamma(\epsilon,\tau).
\end{align*}
Thus $e_T\le\gamma(\epsilon,\tau)$. Since the induction also gives $x^{(T)}\neq0$, time $T$ satisfies the defining condition for $t_{\mathrm{hit}}$, contradicting $t_{\mathrm{hit}}>T$. Hence $t_{\mathrm{hit}}\le T$.

The radius estimate for all $1\le t\le t_{\mathrm{hit}}$ follows from the same induction. This completes the proof.
\end{proof}

%\begin{corollary}[Deterministic stopping under an a-priori corruption budget]
%\label{cor:adversarial-deterministic-stopping}
%Under the hypotheses of \cref{thm:adversarial-positive}, suppose an a-priori upper bound $\tau\le\tau_0$ is known. Define
%\begin{align*}
%	\gamma(\epsilon,\tau_0)
%	\triangleq
%	\epsilon
%	+
%	c_1\sqrt{\epsilon\tau_0}
%	+
%	c_2\tau_0\sqrt{\log\parens*{\frac{2e}{\tau_0}}}
%\end{align*}
%where $c_1, c_2$ are the absolute constants of \cref{thm:adversarial-positive}. If $\gamma(\epsilon,\tau_0)\le\frac{1}{64}$ then \cref{alg:biht}, without per-iterate normalization, when run for $T_0 = 1+\ceil*{\log_2\parens*{\frac{1}{64\gamma(\epsilon,\tau_0)}}}$ iterations, returns an estimate $\widehat{x}_{\mathrm{BIHT}}$ satisfying 
%\begin{align*}
%	\norm{\widehat{x}_{\mathrm{BIHT}}-\xstar}_2
%	\le
%	\gamma(\epsilon,\tau_0)
%\end{align*}
%uniformly over all $\xstar\in\Sigma_s^n\cap S^{n-1}$, all corrupted measurements satisfying $d_{\mathrm H}(y,\sign(A\xstar))\le\tau m$ with $\tau\le\tau_0$, and all sparse-unit initializations.
%\end{corollary}

\begin{corollary}[Deterministic stopping under an a priori corruption budget]
\label{cor:adversarial-deterministic-stopping}
Under the hypotheses of \cref{thm:adversarial-positive}, suppose an apriori upper bound $\tau\le\tau_0$ is known. Define
\begin{align*}
	\gamma(\epsilon,\tau_0)
	\triangleq
	\epsilon
	+
	c_1\sqrt{\epsilon\tau_0}
	+
	c_2\tau_0\sqrt{\log\parens*{\frac{2e}{\tau_0}}},
\end{align*}
where $c_1, c_2$ are the absolute constants of \cref{thm:adversarial-positive}. If $\gamma(\epsilon,\tau_0)\le 1/64$, then, with probability at least $1-\rho$ over the draw of $A$, the following holds simultaneously for every $\xstar\in\Sigma_s^n\cap\mathcal{S}^{n-1}$, every $y\in\{\pm1\}^m$ with $d_{\mathrm H}(y,\sign(A\xstar))\le\tau_0 m$, and every initialization $x^{(0)}\in\Sigma_s^n\cap\mathcal{S}^{n-1}$: \cref{alg:biht}, without per-iterate normalization, when run for $T_0 = 1+\ceil*{\log_2\parens*{\frac{1}{64\gamma(\epsilon,\tau_0)}}}$ iterations, returns an estimate $\widehat{x}_{\mathrm{BIHT}}$ satisfying
\begin{align*}
	\norm{\widehat{x}_{\mathrm{BIHT}}-\xstar}_2
	\le
	\gamma(\epsilon,\tau_0).
\end{align*}
\end{corollary}

\begin{proof}
Let $b_1,b_2,b_3,b_4$, $L$, and $\delta=\epsilon/(c_2L)$ be as in the proof of \cref{thm:adversarial-positive}. We repeat the induction with $\tau_0$ in place of $\tau$ and with the fixed floor $\gamma(\epsilon,\tau_0)$. Since every corruption pattern with the realised corruption fraction $\tau\le\tau_0$ is also admissible at level $\tau_0$, the same adversarial \RAIC{} event applies.

Whenever $\xt\neq0$, define $\rt\triangleq\norm{\xt}_2$, $\ut\triangleq\xt/\norm{\xt}_2$, and $e_t\triangleq\norm{\ut-\xstar}_2$.

The base-case argument in \cref{thm:adversarial-positive}, applied with
$\tau_0$ in place of $\tau$, gives
\begin{align*}
	x^{(1)}\neq0,
	\qquad
	e_1\le\frac{1}{64},
	\qquad
	\abs*{r^{(1)}-1}\le\frac{1}{128}
\end{align*}
We claim that for every $1\le t\le T_0$,
\begin{align*}
	\xt\neq0,
	\qquad
	e_t\le
	\max\set*{
		\frac{2^{-(t-1)}}{64},
		\gamma(\epsilon,\tau_0)
	},
	\qquad
	\abs*{\rt-1}\le\frac{1}{64}
\end{align*}
The base case is proved above. For the inductive step, the same support-restriction and moving-reference argument as in \cref{thm:adversarial-positive} gives
\begin{align*}
	\norm{\xtp-\rt\xstar}_2
	\le
	2 \parens *{\abs*{\rt-1}e_t+b_1\sqrt{\delta e_t}+b_2\delta+b_3\sqrt{\delta\tau_0}+b_4\tau_0\sqrt{\log\parens*{\frac{2e}{\tau_0}}}}
\end{align*}
%\qquad
%	\abs*{r_{t+1}-r_t}\le 2D_t
%\end{align*}

Since $\max\set{e_t, \gamma\parens*{\epsilon, \tau_0}} \ge e_t$ and $\max \set{e_t, \gamma\parens*{\epsilon, \tau_0}} \ge \gamma\parens*{\epsilon, \tau_0}$, for the same choices of $c_1, c_2, L, \delta$ as in \cref{eq:adv-contraction-constant-requirement}, we have
\begin{align*}
\abs*{r^{(t)}-1}e_t
+
b_1\sqrt{\delta e_t}
+
b_2\delta
+
b_3\sqrt{\delta\tau_0}
+
b_4\tau_0\sqrt{\log\parens*{\frac{2e}{\tau_0}}}
&\le 
\parens *{\frac{1}{64}
+
\frac{b_1}{\sqrt{c_2L}}
+
\frac{b_2}{c_2L}
+
\frac{b_3}{c_2\sqrt L}
+
\frac{b_4}{c_2}} \max \set{e_t , \gamma\parens*{\epsilon, \tau_0}}\\
&\le \frac{1-\nicefrac{1}{64}}{16}\max \set{e_t , \gamma\parens*{\epsilon, \tau_0}}
\end{align*}
By assumption $\gamma(\epsilon, \tau_0) \le \frac{1}{64}$. By the induction hypothesis, $e_t \le \max \set {\frac{2^{-(t-1)}}{64}, \gamma\parens*{\epsilon, \tau_0}}$. Then, $\max\set{e_t, \gamma\parens*{\epsilon, \tau_0}} \le 1/64$. Since $\rt \ge 63/64$, the normalization estimate around $\rt\xstar$ applies exactly as before and yields
\begin{align*}
	e_{t+1}
	\le
	\frac{1}{2} \max\set{e_t, \gamma\parens*{\epsilon, \tau_0}}
	\le
	\max\set*{
		\frac{2^{-t}}{64},
		\gamma(\epsilon,\tau_0)
	}
\end{align*}
For the radius, summing the one-step radial estimate gives
\begin{align}
\abs*{\rtp-1}
&\le
\frac{1}{128}
+
\frac{1-\nicefrac{1}{64}}{8}
\sum_{j=1}^{T_0-1}
\max\set*{
	\frac{2^{-(j-1)}}{64},
	\gamma(\epsilon,\tau_0)
}\label{eq:radial-sum}
\end{align}
The contribution from the first term is at most $1/32$. For the second term, let $q=1/(64\gamma(\epsilon,\tau_0))\ge1$. Then, 
\begin{align*}
	(T_0-1)\gamma(\epsilon,\tau_0)
	=
	\ceil*{\log_2 q}\gamma(\epsilon,\tau_0)
	\le
	2q\gamma(\epsilon,\tau_0)
	=
	\frac{1}{32}
\end{align*}
where the last inequality uses $\log_2 q \le q$ together with $q \ge 1$. Then, the sum in the second term of \cref{eq:radial-sum} is at most $1/16$, and 
\begin{align*}
	\abs*{\rtp-1}
	\le
	\frac{1}{128}
	+
	\frac{1-\nicefrac{1}{64}}{128}
	<
	\frac{1}{64}
\end{align*}
This closes the induction. At time $T_0$,
\begin{align*}
	e_{T_0}
	\le
	\max\set*{
		\frac{2^{-(T_0-1)}}{64},
		\gamma(\epsilon,\tau_0)
	}
	\le
	\gamma(\epsilon,\tau_0)
\end{align*}
by the definition of $T_0$, and that concludes the proof. 
\end{proof}
\begin{corollary}[Stable window after deterministic stopping] \label{cor:adversarial-stable-window}
Under the hypotheses and notation of \cref{cor:adversarial-deterministic-stopping}, suppose the same fixed-step iterates are continued past $T_0$. Then, on the same high-probability event as in \cref{cor:adversarial-deterministic-stopping}, for every
\begin{align*}
	T_0\le t\le T_0+\floor*{\frac{1}{5\gamma(\epsilon,\tau_0)}}
\end{align*}
we have $x^{(t)}\neq0$,
\begin{align*}
	\abs*{\norm{x^{(t)}}_2-1}\le\frac{1}{16}
\end{align*}
and
\begin{align*}
	\norm*{
		\frac{x^{(t)}}{\norm{x^{(t)}}_2}
		-
		\xstar
	}_2
	\le
	\gamma(\epsilon,\tau_0)
\end{align*}
with the same uniformity over targets, admissible corruptions, and sparse-unit initializations.
\end{corollary}

\begin{proof}
The proof of \cref{cor:adversarial-deterministic-stopping} gives
\begin{align*}
	x^{(T_0)}\neq0,
	\qquad
	e_{T_0}\le\gamma(\epsilon,\tau_0),
	\qquad
	\abs*{r^{(T_0)}-1}\le\frac{1}{64}
\end{align*}
where, whenever $\xt\neq0$, we write $\rt=\norm{\xt}_2$, $\ut=\xt/\norm{\xt}_2$, and $e_t=\norm{\ut-\xstar}_2$.

We prove by induction that for every
\begin{align*}
	T_0\le t\le T_0+\floor*{\frac{1}{5\gamma(\epsilon,\tau_0)}}
\end{align*}
we have
\begin{align*}
	\xt\neq0,
	\qquad
	e_t\le\gamma(\epsilon,\tau_0),
	\qquad
	\abs*{\rt-1}\le\frac{1}{16}
\end{align*}
The base case follows from \cref{cor:adversarial-deterministic-stopping}. Assume the claim holds at some time $t$ such that
\begin{align*}
	T_0\le t<T_0+\floor*{\frac{1}{5\gamma(\epsilon,\tau_0)}}
\end{align*}
The same support-restriction and moving-reference argument as in \cref{thm:adversarial-positive} gives
\begin{align*}
\norm{\xtp-\rt\xstar}_2
\le
2\parens*{
\abs*{\rt-1}e_t
+
b_1\sqrt{\delta e_t}
+
b_2\delta
+
b_3\sqrt{\delta\tau_0}
+
b_4\tau_0\sqrt{\log\parens*{\frac{2e}{\tau_0}}}
}
\end{align*}
Using $e_t\le\gamma(\epsilon,\tau_0)$, $\abs*{\rt-1}\le1/16$, and the same termwise bounds as in the proof of \cref{cor:adversarial-deterministic-stopping}, we obtain
\begin{align*}
&\abs*{\rt-1}e_t
+
b_1\sqrt{\delta e_t}
+
b_2\delta
+
b_3\sqrt{\delta\tau_0}
+
b_4\tau_0\sqrt{\log\parens*{\frac{2e}{\tau_0}}} 
&\le
\parens*{
\frac{1}{16}
+
\frac{b_1}{\sqrt{c_2L}}
+
\frac{b_2}{c_2L}
+
\frac{b_3}{c_2\sqrt L}
+
\frac{b_4}{c_2}
}
\gamma(\epsilon,\tau_0)
\end{align*}
By \cref{eq:adv-contraction-constant-requirement},
\begin{align*}
\frac{b_1}{\sqrt{c_2L}}
+
\frac{b_2}{c_2L}
+
\frac{b_3}{c_2\sqrt L}
+
\frac{b_4}{c_2}
\le
\frac{1-\nicefrac{1}{64}}{16}
-
\frac{1}{64}
=
\frac{47}{1024}
\end{align*}
Therefore,
\begin{align*}
\abs*{\rt-1}e_t
+
b_1\sqrt{\delta e_t}
+
b_2\delta
+
b_3\sqrt{\delta\tau_0}
+
b_4\tau_0\sqrt{\log\parens*{\frac{2e}{\tau_0}}}
\le
\frac{111}{1024}\gamma(\epsilon,\tau_0)
\end{align*}
Since $\gamma(\epsilon,\tau_0)\le1/64$ and $\rt\ge15/16$, the normalization estimate around $\rt\xstar$ applies exactly as before. Then, 
\begin{align*}
e_{t+1}
&\le
\frac{8}{\rt}\cdot
\frac{111}{1024}\gamma(\epsilon,\tau_0)\\
&\le
\frac{128}{15}\cdot
\frac{111}{1024}\gamma(\epsilon,\tau_0)\\
&=
\frac{111}{120}\gamma(\epsilon,\tau_0)
<
\gamma(\epsilon,\tau_0)
\end{align*}
This proves the directional part of the induction.

For the radius, the same one-step estimate gives
\begin{align*}
\abs*{r^{(t+1)}-\rt}
&\le
\norm{\xtp-\rt\xstar}_2\\
&\le
\frac{111}{512}\gamma(\epsilon,\tau_0)
\end{align*}
Thus, for every $t$ in the claimed window,
\begin{align*}
\abs*{\rt-1}
&\le
\abs*{r^{(T_0)}-1}
+
\sum_{j=T_0}^{t-1}
\abs*{r^{(j+1)}-r^{(j)}}\\
&\le
\frac{1}{64}
+
\frac{111}{512}\gamma(\epsilon,\tau_0)
\floor*{\frac{1}{5\gamma(\epsilon,\tau_0)}}\\
&\le
\frac{1}{64}
+
\frac{111}{2560}=
\frac{151}{2560}
<
\frac{1}{16}
\end{align*}
This closes the induction and completes the proof.
\end{proof}

\section{Discussion}
The lower bound under corruptions is stated for a fixed step-size. It is natural to ask whether variable step-sizes can smooth out the scalar oscillation. In \cref{app:loss-landscape}, we revisit this question from the loss-landscape perspective and find that variable step-sizes do not give a genuinely new last-iterate mechanism. Objective-driven choices fail, with Polyak's step-size pushing every non-zero scalar iterate towards the origin. For time-only schedules, unless the total future motion is finite, the scalar oscillation persists. If the total motion is finite, the movement budget must be placed strategically, which requires apriori corruption information and essentially reduces to a softened version of the stopping rule in \cref{cor:adversarial-deterministic-stopping}. Norm-dependent schedules can avoid this obstruction by controlling the relative step-size seen by the direction, but that just reproduces normalized \BIHT{} at the level of directions. 

Without noise, our all-time guarantee uses more measurements for targets whose nonzero entries are below $\eps$, and \cref{thm:noiseless-late-failure} shows that an argument based on RAIC alone cannot avoid this. Whether Gaussian measurements keep normalization-free BIHT $\eps$-accurate at every iteration after $T$, for every target and every initialization, with $\widetilde{O}(s/\eps)$ measurements, remains open.

The remainder of our analysis, like much of the theory for one-bit compressed sensing, is stated for the Gaussian measurement model. The main probabilistic input is the restricted approximate invertibility condition, which uses the angular geometry of Gaussian signs. Establishing RAIC for other sensing matrices remains an interesting direction for future work, and any such matrix would inherit all our results verbatim.

\section*{Acknowledgements}

We thank Namiko Matsumoto for her helpful comments on an earlier version. This work is supported by NSF TRIPODS award 2217058 (EnCORE Inst.) and 2112665.

\section*{Statement of AI use}
All ideas and technical contributions in this paper are the authors' own. Claude Fable 5 and Opus 5.5 were used for language editing, proofreading, and making small changes to the code from \textcite{pmlr-v291-matsumoto25a} for the experiments in \cref{fig:numerics}. The authors bear sole responsibility for the content and correctness of this manuscript.

\PrintBibliography

\newpage
\appendix

\appendix

\crefname{appendix}{appendix}{appendices}
\Crefname{appendix}{Appendix}{Appendices}

\crefalias{section}{appendix}
\crefalias{subsection}{appendix}
\crefalias{subsubsection}{appendix}

\setcounter{section}{0}
\renewcommand{\thesection}{A.\arabic{section}}
\renewcommand{\theHsection}{appendix.\arabic{section}}

\section{Numerical Ledger}
We specify our choice of constants for our main results. For \cref{thm:noiseless}, $c_1 = \sqrt{\frac{\pi}{(380\times 512)}}\parens *{\sqrt{3}+16} \approx 0.07125$ and $c_2 = \frac{90\sqrt{2}}{380} \approx 0.3349$ are the absolute constants in \cref{fact:gaussian-matrices-satisfy-raic}, as they appear in \textcite{Matsumoto_noiseless_2024}, and $L = 13$. For \cref{thm:adversarial-positive}, $b_1 = 1.347, b_2 = 0.3807, b_3 = 1.2501$ and $b_4 = 9.0899$ are the universal constants in \cref{fact:gaussian-matrices-satisfy-adv-raic}. These constants dominate the corresponding constants in Theorem 3.1 of \textcite{Matsumoto_robust_2023} and by monotonicity of \RAIC{}, the same conclusion holds. It suffices to choose $c_1 = 16, c_2 = c_1 ^2 = 256$ and $L=128$ to ensure the conditions of \cref{eq:adv-base-constant-requirement} and \cref{eq:adv-contraction-constant-requirement}. 

\section{The loss landscape and step-size schedules}\label{app:loss-landscape}
The lower bound in \cref{thm:adversarial-lower-bound} is stated for the original fixed-step \BIHT{} update. It is natural to ask whether the obstruction is an artifact of keeping the step-size fixed. We find that this is not the case. 

Consider the sign-consistency loss
\begin{align*}
    l_y(x)
    \triangleq
    \frac{\eta}{m}\sum_{i=1}^{m}
    \max\parens*{0,-y_i a_i^\top x}.
\end{align*}
Away from the measurement hyperplanes $a_i^\top x=0$, the BIHT correction is the negative gradient of this loss, i.e.,
\begin{align*}
    -\nabla l_y(x)
    =
    \frac{\eta}{m}A^\top\frac{1}{2}
    \parens*{y-\sign(Ax)}.
\end{align*}
It is easy to see that $l_y(\kappa x) = \kappa l_y(x)$ for all $\kappa \ge 0$, and the corrupted objective lives on the sparse cone $\Sigma_s^n$. Moreover, if $x\ne0$ the loss may be decomposed into radial and angular components as 
\begin{align*}
    l_y(x) = \norm{x}_2 \, l_y\parens*{x/\norm{x}_2}.
\end{align*}
%Per-iterate normalization removes the radial error component by restricting the iterates to the sparse unit sphere. 
With $\alpha, \gamma > 0$ defined as before, the corrupted loss is
\begin{align*}
    l_y(x) = \alpha\max\set{x,0} + \gamma\max\set{-x,0}
\end{align*}
The unique minimizer on the cone is $x=0$, and the negative gradient points towards the origin at every iterate, even if it has the correct sign. On the sphere, $S^0 = \set{-1, +1}$, and the same loss has form
\begin{align*}
    l_y(x)=
    \begin{cases}
        \alpha & \text{ if }x=1\\
        \gamma & \text{ if }x=-1
    \end{cases}
\end{align*}
Whenever $\alpha < \gamma$, the loss prefers the correct direction $+1$. Taken together, these explain the scalar obstruction in \cref{thm:adversarial-lower-bound}.

Suppose we modify the \BIHT{} iterate with a variable step-size $\lambda_t$, as follows:
\begin{align*}
    \xtp \in \cT_s\parens *{\xt+\lambda_t h_{f;A}\parens *{x^\star,\xt}}
\end{align*}
\paragraph{Objective-driven step-sizes. } Common objective-driven step-sizes do not avoid this instance. Polyak's step-size formula with optimal $l_y^\star = 0$ is of the form
\begin{align*}
    \lambda_t = \frac{l_y\parens{\xt}}{\abs*{ l^{'}_y\parens{\xt}}^2}.
\end{align*}
This gives $\lambda_t = \xt/\alpha$ when $\xt > 0$ and $\lambda_t = -\xt/\gamma$ when $\xt < 0$, and the next iterate is $0$ either way. 
\paragraph{Time-only step sizes. }Suppose $\lambda_t$ is deterministic and depends only on time. If $\sum_{t=0}^{\infty}\lambda_t=\infty$ then the obstruction in \cref{thm:adversarial-lower-bound} survives unchanged. The amplitude of the oscillation may approach zero for slowly decaying schedules, but the trajectory is dependent on the sum $\sum_{t \ge T}\lambda_t$, which forces the normalized direction to change signs infinitely often. 

Such a schedule can avoid the scalar obstruction only if it has finite total mass. This condition is also not sufficient by itself, unless the finite budget is placed correctly. Starting from the correct scalar sign $x^{(0)}=1$, the trajectory remains positive for all time provided
\begin{align*}
    1-\alpha\sum_{j=0}^{t-1}\lambda_j>0
    \qquad\text{for every } t\geq 1 .
\end{align*}
Starting from the wrong scalar sign $x^{(0)}=-1$, the schedule must first spend enough mass to cross to the positive side. If there is no $t\geq 1$ such that
\begin{align*}
    \gamma\sum_{j=0}^{t-1}\lambda_j>1,
\end{align*}
then the wrong sign is never corrected. If $t$ is the first time for which the above inequality holds, then that sign is preserved by the remaining tail if 
\begin{align*}
    \alpha\sum_{j=t}^{\infty}\lambda_j
    <
    \gamma\sum_{j=0}^{t-1}\lambda_j-1 .
\end{align*}
Therefore, a summable schedule can only succeed by placing enough weight on the early iterates to correct a wrong sign, and small weight on the remaining iterates so as to not overturn the recovered sign. Certifying such a budget depends on the scalar pushes $\alpha$ and $\gamma$, or on an a-priori upper bound on the corruption level, which then boils down to a strategy similar to that of \cref{cor:adversarial-deterministic-stopping}: run a fixed-step dynamics until the certified recovery time, followed by a summable tail schedule that appends a controlled amount of residual motion.

\paragraph{Norm-dependent step size. } The remaining possibility is to make the step-size depend on the current norm. The scalar recursion suggests that the relevant quantity is the relative step-size $\lambda_t / \rt$, with $\rt = \abs{\xt}$. For $\xt > 0$,
\begin{align*}
    \xtp = \rt\parens *{1-\alpha\frac{\lambda_t}{\rt}}
\end{align*}
while, if $\xt < 0$, then
\begin{align*}
    \xtp = \rt\parens*{-1 + \gamma\frac{\lambda_t}{\rt}}
\end{align*}
This necessitates 
\begin{align*}
    \frac{1}{\gamma} \stackrel{(1)}{<} \frac{\lambda_t}{\rt} \stackrel{(2)}{<} \frac{1}{\alpha}
\end{align*}
where (1) ensures that a wrong negative sign is corrected and (2) preserves a correct positive sign. Both requirements can hold when the clean push dominates the corrupted push, i.e. $\alpha<\gamma$. Moreover, if the rule is meant to remain corrective at arbitrarily small radii without freezing, then $\lambda_t$ must scale linearly with the current radius. If $\lambda_t\ll\rt$, then the effective normalized step-size vanishes and the dynamics freeze. If $\lambda_t\gg\rt$, then even a correct small iterate can be pushed across the origin. The correct choice is therefore $\lambda_t = \kappa_t\rt$ where $\kappa_t>0$ may be chosen adaptively. By homogeneity of thresholding, the next iterate has the form
\begin{align*}
\xtp\in\cT_s\parens*{\rt\ut+\lambda_t h_{f;A}\parens*{\xstar,\ut}} =
\rt\,\cT_s\parens*{\ut+\kappa_t h_{f;A}\parens*{\xstar,\ut}} .
\end{align*}
Whenever the next iterate is nonzero,
\begin{align*}
\frac{\xtp}{\norm{\xtp}_2}
\in
\frac{\cT_s\parens*{\ut+\kappa_t h_{f;A}\parens*{\xstar,\ut}}}
{\norm{\cT_s\parens*{\ut+\kappa_t h_{f;A}\parens*{\xstar,\ut}}}_2}.
\end{align*}
This is precisely the normalized \BIHT{} direction update with step-size multiplier $\kappa_t$. In particular, $\kappa_t=1$ gives the usual normalized \BIHT{} direction sequence.

\section{Deferred proofs : late iterate behavior of noiseless BIHT}\label{app:defer_late_iterate_noiseless}

\begin{theorem}[All-time guarantee]
\label{app:thm:noiseless-all-time}
Fix $0<\mu\le1$. Assume the hypotheses of \cref{thm:noiseless}, with $\eps$ replaced by $\min\{\eps,\mu\}$ in Eq.~\eqref{eq:suff_measurements}, and let $T=1+\ceil{\log_2(1/(64\eps))}$. Then, with probability at least $1-\rho$ over the draw of $A$, the following holds simultaneously for every target $\xstar\in\Sigma_s^n\cap S^{n-1}$ whose nonzero entries have magnitude at least $\mu$ and every sparse-unit initialization: the normalization-free BIHT iterates satisfy $\xt\neq0$ and
\begin{align*}
    \norm*{\frac{\xt}{\norm{\xt}_2}-\xstar}_2\le\eps
    \qquad\text{for every integer }t\ge T
\end{align*}
In particular, when $\mu\ge\eps$, the measurement requirement is the same as in \cref{thm:noiseless}.
\end{theorem}

\begin{proof}
We reuse the notation and constants from the proof of \cref{thm:noiseless}. The assumed measurement bound gives RAIC at scale $\min\{\eps,\mu\}/L$ with probability at least $1-\rho$. We condition on this event.

Suppose first that $\mu\ge\eps$. Then the RAIC resolution is $\eps/L$, and every nonzero target coordinate has magnitude at least $\eps$. We first note that whenever
$\xt\neq0$ and $e_t\le\eps$, the iterate contains every target coordinate with the correct sign. Indeed, if $u_i^{(t)}x_i^\star\le0$ for some $i\in\supp(\xstar)$, then
\begin{align*}
    e_t
    \ge \abs{u_i^{(t)}-x_i^\star}
    =\abs{u_i^{(t)}}+\abs{x_i^\star}
    \ge\abs{x_i^\star}
    \ge\eps
\end{align*}
Since $e_t\le\eps$, equality must hold throughout. This forces $u_i^{(t)}=0$, $\abs{x_i^\star}=\eps$, and $u_j^{(t)}=x_j^\star$ for every $j\neq i$, contradicting $\norm{\ut}_2=\norm{\xstar}_2=1$. Therefore, $\supp(\xstar)\subseteq\supp(\xt)$, with matching signs on the target support.

The proof of \cref{thm:noiseless} gives $x^{(T)}\neq0$, $e_T\le\eps$, and $\abs{r^{(T)}-1}\le1/64$. Consequently,
\begin{align*}
    \dotprod{x^{(T)},\xstar}
    =r^{(T)}\left(1-\frac{e_T^2}{2}\right)
    \ge\frac{63}{64}\left(1-\frac{\eps^2}{2}\right)
    >\frac{15}{16}
\end{align*}
We prove by induction that, for every $t\ge T$,
\begin{align*}
    \xt\neq0,\qquad
    e_t\le\eps,\qquad
    \dotprod{\xt,\xstar}
    \ge\dotprod{x^{(T)},\xstar}
\end{align*}
The base case holds by the preceding bounds. Suppose the claim holds at time $t$. By Cauchy--Schwarz,
\begin{align*}
    \rt\ge\dotprod{\xt,\xstar}>\frac{15}{16}
\end{align*}

If $\rt\le1$, then $\abs{\rt-1}\le1/16$. The moving-reference estimate from \cref{thm:noiseless} gives
\begin{align*}
    \norm{\xtp-\rt\xstar}_2
    &\le
    2\left(\frac1{16}+\frac{c_1}{\sqrt L}
        +\frac{c_2}{L}\right)\eps\\
    &\le\frac{111}{512}\eps<\frac{\rt}{2}
\end{align*}
Moreover, \cref{lem:normalization-around-radial-target} gives $\xtp\neq0$ and
\begin{align*}
    e_{t+1}\le\frac{111}{120}\eps<\eps
\end{align*}
and by our previous observation, $\supp(\xstar)\subseteq\supp(\xtp)$.

Suppose instead that $\rt>1$. As in \cref{thm:noiseless}, the next iterate is a hard-thresholding of the support-restricted dense update
\begin{align*}
    \xt+h_{A;J_{t+1}}(\xstar,\ut)
    &=(\rt-1)\ut+\xstar+\xi_t
\end{align*}
where RAIC gives
\begin{align*}
    \norm{\xi_t}_2
    \le\left(\frac{c_1}{\sqrt L}+\frac{c_2}{L}\right)\eps
    \le\frac{\eps}{16}
\end{align*}
For every $i\in\supp(\xstar)$, the coordinates $u_i^{(t)}$ and $x_i^\star$ have the same sign. Therefore,
\begin{align*}
    \abs{(\rt-1)u_i^{(t)}+x_i^\star+\xi_{t,i}}
    \ge\abs{x_i^\star}-\abs{\xi_{t,i}}
    \ge\frac{15\eps}{16}
\end{align*}
Every coordinate outside $\supp(\xt)$, however, has magnitude at most $\eps/16$. Since $\supp(\xt)$ contains the target support and has size at most $s$, hard thresholding retains every target coordinate. In particular, $\xtp\neq0$.

By \cref{lem:normalization-around-radial-target},
\begin{align*}
    \norm*{
        \frac{\xstar+\xi_t}{\norm{\xstar+\xi_t}_2}
        -\xstar
    }_2
    \le4\norm{\xi_t}_2\le\frac{\eps}{4}
\end{align*}
Therefore, both $\ut$ and the direction of $\xstar+\xi_t$ have error at most $\eps$. Since $\rt-1>0$ and hard thresholding retains the target support,
\begin{align*}
    \dotprod{\xtp,\xstar}
    &=(\rt-1)\dotprod{\ut,\xstar}
      +\dotprod{\xstar+\xi_t,\xstar}\\
    &\ge
      \left(1-\frac{\eps^2}{2}\right)
      \left(\rt-1+\norm{\xstar+\xi_t}_2\right)\\
    &\ge
      \left(1-\frac{\eps^2}{2}\right)\norm{\xtp}_2
\end{align*}
The last inequality uses the fact that hard thresholding cannot increase the norm, followed by the triangle inequality. Dividing by $\rtp$ gives $e_{t+1}\le\eps$.

In both cases, the next iterate is nonzero, has error at most $\eps$, and retains every target coordinate. Noiselessness therefore gives
\begin{align*}
    \dotprod{\xtp,\xstar}
    &=\dotprod{\xt,\xstar}
      +\dotprod{h_{A;J_{t+1}}(\xstar,\ut),\xstar}\\
    &\ge\dotprod{\xt,\xstar}
\end{align*}
This closes the induction when $\mu\ge\eps$.

Finally, suppose $\mu<\eps$. On the chosen RAIC event, we may apply the preceding argument with $\mu$ in place of $\eps$. It gives nonzero iterates and $e_t\le\mu<\eps$ for every $t\ge1+\ceil{\log_2(1/(64\mu))}$. The BIHT update does not depend on the requested accuracy, so this is the same sequence of iterates. For every $T\le t<1+\ceil{\log_2(1/(64\mu))}$, the proof of \cref{thm:noiseless} at accuracy $\mu$ gives nonzero iterates and
\begin{align*}
    e_t\le\max\set*{\frac{2^{-(t-1)}}{64},\mu}\le\eps
\end{align*}
where the last inequality follows from $t\ge T$ and $\mu<\eps$.
\end{proof}

\begin{theorem}[Late failure under RAIC]
\label{app:thm:noiseless-late-failure}
Let $0<\eps\le1/64$, and let $L,c_1,c_2$ be as in the proof of \cref{thm:noiseless}. Set $\delta=\eps/L$ and $T=1+\ceil{\log_2(1/(64\eps))}$.

For every sufficiently large $m$, there exist $A\in\R^{m\times3}$ satisfying $(2,3,\delta,c_1,c_2)$-RAIC and $\xstar,x^{(0)}\in\Sigma_2^3\cap S^2$ with
$\min_{i\in\supp(\xstar)}\abs{x_i^\star}<\eps$ such that the following holds. With sparsity $s=2$ and noiseless measurements $y=\sign(A\xstar)$, normalization-free BIHT has a finite time $\tau>T+\floor*{1/(64\eps^2)}$ for which all iterates through $\tau$ are nonzero and 
\begin{align*}
    e_t<\eps\quad(0\le t<\tau),
    \qquad
    e_\tau>\frac12
\end{align*}
\end{theorem}
\begin{proof}
Let $B\in\R^{(m-2)\times3}$ have independent $\cN(0,I_3)$ rows. We take $m$ sufficiently large that \cref{fact:gaussian-matrices-satisfy-raic} applies to $B$ at resolution $\delta/4$ with failure probability $1/4$, and
\begin{align*}
    m>\frac1\eps,\qquad
    \frac{\eta}{m}\le\frac14,\qquad
    \frac{2\eta}{m}\le\frac{c_2\delta}{2}
\end{align*}
Fix $\theta=1/(8m)$. By rotational invariance of Gaussians, $B_i/\norm{B_i}_2$ is uniform on $S^2$, and
\begin{align*}
    \Pr{\exists i\in[m-2]:
    \abs{B_{i1}}\le2\theta\norm{B_i}_2}
    \le2\theta(m-2)<\frac14
\end{align*}
Consequently, with probability at least $1/2$, there exists a realization of $B$ satisfying $(2,3,\delta/4,c_1,c_2)$-RAIC and
\begin{align*}
    \abs{B_{i1}}>2\theta\norm{B_i}_2
    \qquad\forall i\in[m-2]
\end{align*}
Fix such a realization and define
\begin{align*}
    A&=
    \begin{pmatrix}
        \multicolumn{3}{c}{\frac{m}{m-2}B}\\
        \theta^2/4&-\theta/2&0\\
        \theta&0&-1
    \end{pmatrix}\\
    \xstar&=\frac{(1,\theta,0)}{\sqrt{1+\theta^2}},
    \qquad
    x^{(0)}=
    \left(\sqrt{1-\frac{\eta^2\theta^2}{m^2}},
    0,\frac{\eta\theta}{m}\right)
\end{align*}
Both vectors are $2$-sparse and have unit norm. Moreover,
\begin{align*}
    \min_{i\in\supp(\xstar)}\abs{x_i^\star}
    =\frac{\theta}{\sqrt{1+\theta^2}}<\eps
\end{align*}

We first verify RAIC for $A$. The scaling of $B$ ensures that its contribution to $h_{A;J}(p,q)$ is exactly $h_{B;J}(p,q)$. Since the sum of the norms of the two appended rows is less than $2$, for every admissible $p,q,J$, we have
\begin{align*}
    \norm{(p-q)-h_{A;J}(p,q)}_2
    &\le\norm{(p-q)-h_{B;J}(p,q)}_2
        +\frac{2\eta}{m}\\
    &\le\frac{c_1}{2}
        \sqrt{\delta\norm{p-q}_2}
        +\frac{c_2\delta}{4}
        +\frac{2\eta}{m}\\
    &\le c_1\sqrt{\delta\norm{p-q}_2}
        +c_2\delta
\end{align*}
Thus $A$ satisfies $(2,3,\delta,c_1,c_2)$-RAIC.

Next, the choice of $B$ implies that, whenever $u^2+v^2\le4\theta^2$,
\begin{align*}
    \abs{\dotprod{B_i,(1,u,v)}-B_{i1}}
    \le2\theta\norm{B_i}_2<\abs{B_{i1}}
\end{align*}
Therefore, the measurements from $B$ have the same signs on $\xstar$ and every positive multiple of $(1,0,v)$ with $0\le v\le2\theta$.

We argue that the iterates have this form for every $0\le t <\tau$. In particular, we will establish that, starting from $x^{(0)}$, the third coordinate stays fixed at $\eta\theta/m$, while the first coordinate decreases at every step. Eventually, the coordinate becomes so small that an additional row of $A$ is mismatched, which causes a sharp spike in the directional error. 

This holds trivially at $t=0$; since $\eta/m\le1/4$ and $\theta\le1/8$, the third coordinate of $x^{(0)}$ is at most $1/32$, so $x_1^{(0)}=\sqrt{1-(\eta\theta/m)^2}>3/4>\eta/m$.

Suppose that, at some iteration $t$, we have
\begin{align*}
    x^{(t)}
    =\left(x_1^{(t)},0,\frac{\eta\theta}{m}\right),
    \qquad x_1^{(t)}\ge\frac{\eta}{m}
\end{align*}
Then $x^{(t)}$ is a positive multiple of $(1,0,v)$ with
\begin{align*}
    0<v=\frac{\eta\theta}{mx_1^{(t)}}\le\theta
\end{align*}
so all rows of $B$ agree with the target. The last row of $A$ remains correctly classified, since its target label is $+1$ and $\langle a_m, x^{(t)}\rangle = \theta\parens{x_1^{(t)}-\eta/m}\ge 0$. The penultimate row, however, is mismatched, since  $\langle a_{m-1}, x^{(t)}\rangle = \theta^2x_1^{(t)}/4>0$ while its target label is $-1$. Therefore, this is the only row that contributes to the correction, and the dense iterate at $t+1$ is 
\begin{align*}
    x^{(t)}+h_A(\xstar,x^{(t)})
    =\left(
        x_1^{(t)}-\frac{\eta\theta^2}{4m},
        \frac{\eta\theta}{2m},
        \frac{\eta\theta}{m}
    \right)
\end{align*}
Moreover,
\begin{align*}
    x_1^{(t)}-\frac{\eta\theta^2}{4m}
    &\ge\frac{\eta}{m}\left(1-\frac{\theta^2}{4}\right)
    >\frac{\eta\theta}{m}
    >\frac{\eta\theta}{2m}
\end{align*}
Hard thresholding therefore retains coordinates $1$ and $3$, which gives
\begin{align*}
    x^{(t+1)}
    =\left(
        x_1^{(t)}-\frac{\eta\theta^2}{4m},
        0,\frac{\eta\theta}{m}
    \right)
\end{align*}

Now, since the first coordinate decreases by exactly $\eta\theta^2/(4m) > 0$ at every iteration, it will eventually fall below $\eta/m$. Let $\tau-1$ denote the first iteration at which this happens.

For every $0\le t <\tau$, the preceding argument gives
\begin{align*}
    x^{(t)}
    =\parens*{
        \sqrt{1-\frac{\eta^2\theta^2}{m^2}}
        -\frac{t\eta\theta^2}{4m},
        0,\frac{\eta\theta}{m}}
\end{align*}
By definition of $\tau-1$, the first coordinate at that iteration satisfies
\begin{align}\label{eq:tau-crossing}
    \frac{\eta}{m}\left(1-\frac{\theta^2}{4}\right)
    \le x_1^{(\tau-1)}
    <\frac{\eta}{m}
\end{align}
In particular, for every $0\le t<\tau$,
\begin{align*}
    x_1^{(t)}>\frac{\eta}{2m},
    \qquad
    0<\frac{\eta\theta}{mx_1^{(t)}}<2\theta
\end{align*}
Therefore, every iterate preceding $\tau$ is a positive multiple of $(1,0,v)$ for some $0<v<2\theta$, as claimed, and all rows of $B$ remain correctly classified. 

This also suggests that these iterates are accurate. For any $0\le t<\tau$, we can write $x^{(t)}$ as a positive multiple of $(1,0,v)$, and by the triangle inequality,
\begin{align*}
    e_t
    &\le
    \norm*{\frac{(1,0,v)}{\sqrt{1+v^2}}-(1,0,0)}_2
    +\norm{\xstar-(1,0,0)}_2\\
    &\le v+\theta
    <3\theta = \frac{3}{8m} < \eps
\end{align*}
Next, we examine the update at iteration $\tau-1$. Since $x_1^{(\tau-1)}<\eta/m$, the last row gives
\begin{align*}
    \langle a_m, x^{(\tau-1)} \rangle = \theta\parens*{x_1^{(\tau-1)}-\frac{\eta}{m}}<0
\end{align*}
and is therefore mismatched, along with the penultimate row. Since all rows of $B$ remain correctly classified, the dense update is
\begin{align*}
    x^{(\tau-1)}+h_A(\xstar,x^{(\tau-1)})
    =\left(
        x_1^{(\tau-1)}
        +\frac{\eta}{m}\left(\theta-\frac{\theta^2}{4}\right),
        \frac{\eta\theta}{2m},
        -\frac{\eta}{m}(1-\theta)
    \right)
\end{align*}
Together with \cref{eq:tau-crossing}, this suggests that the first coordinate of the dense update satisfies
\begin{align*}
    \frac{\eta}{m}
    &<
    x_1^{(\tau-1)}
    +\frac{\eta}{m}\left(\theta-\frac{\theta^2}{4}\right)
    <\frac{\eta}{m}(1+\theta)
\end{align*}
while the third coordinate has magnitude
\begin{align*}
    \frac{\eta}{m}(1-\theta)
    >\frac{\eta\theta}{2m}
\end{align*}
Thus, hard thresholding again retains coordinates $1$ and $3$, and
\begin{align*}
    x^{(\tau)}
    =\left(
        x_1^{(\tau-1)}
        +\frac{\eta}{m}\left(\theta-\frac{\theta^2}{4}\right),
        0,-\frac{\eta}{m}(1-\theta)
    \right)
\end{align*}
However, the third coordinate has now become comparable in magnitude to the first. Since $x_3^\star=0$, we have
\begin{align*}
    e_\tau
    &\ge\frac{\abs{x_3^{(\tau)}}}{\norm{x^{(\tau)}}_2}\\
    &>
    \frac{1-\theta}
    {\sqrt{(1+\theta)^2+(1-\theta)^2}}
    >\frac12
\end{align*}
where the last inequality follows from $\theta\le1/8$. All the thresholding comparisons above are strict, so the trajectory is independent of tie-breaking. Moreover, every iterate through $\tau$ is nonzero.

Finally, we verify that the failure occurs beyond the stated window. Since $x_1^{(0)}>3/4$ and $\eta/m\le1/4$, the definition of $\tau$ gives
\begin{align*}
    \tau
    &>
    \frac{x_1^{(0)}-\eta/m}{\eta\theta^2/(4m)}
    >\frac{2m}{\eta\theta^2}
    =\frac{128m^3}{\eta} = \Omega(m^3)
\end{align*}
On the other hand, $m>1/\eps$ and $T=O(\log(1/\eps))$ give
\begin{align*}
    T+\floor*{\frac1{64\eps^2}}
    =O(m^2)=o(\tau)
\end{align*}
Therefore, for all sufficiently large $m$,
\begin{align*}
    \tau>T+\floor*{\frac1{64\eps^2}}
\end{align*}
which completes the proof.
\end{proof}

\begin{figure}[t]
    \centering

    \begin{subfigure}{\linewidth}
        \centering
        \includegraphics[width=\linewidth]{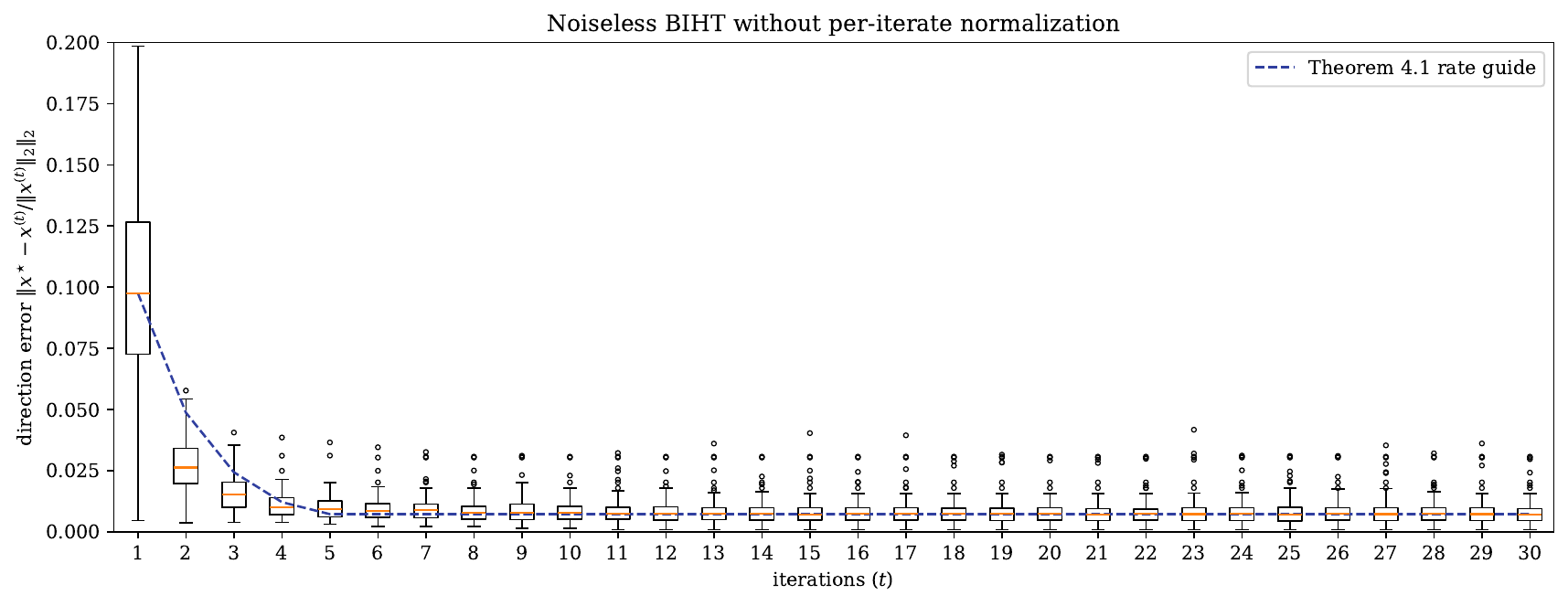}
        \caption{
        %Noiseless measurements. Boxplots show the directional error of \BIHT{}\cite{Jacques2011Robust1C} for $t=1,\ldots,30$ over $100$ independent trials. The dashed curve is the (constant-suppressed) error envelope from \cref{thm:noiseless}, with $\epsilon = 0.01$. 
        Noiseless measurements: BIHT \cite{jacques2013robust}, that is, \cref{alg:biht} without per-iterate normalization.
        }
        \label{fig:numerics-noiseless}
    \end{subfigure}

    \vspace{1.25em}

    \begin{subfigure}{\linewidth}
        \centering
        \includegraphics[width=\linewidth]{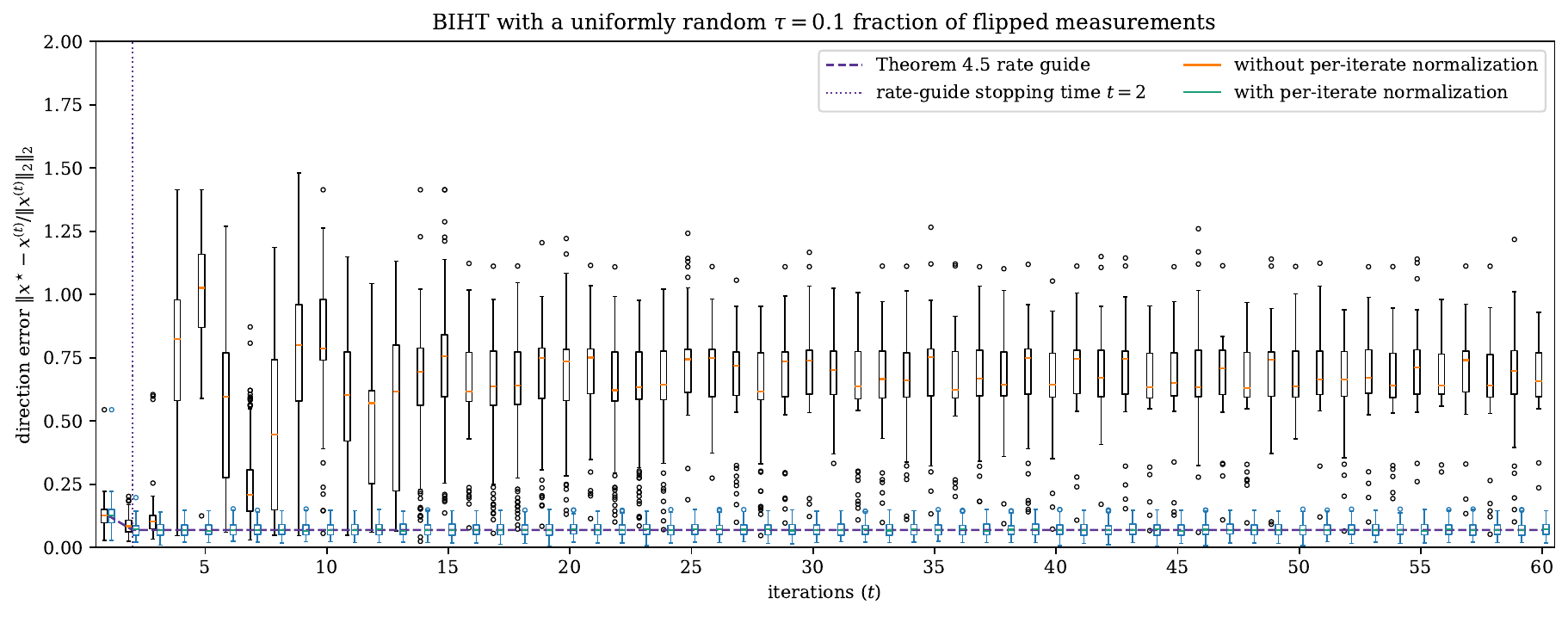}
        \caption{
        %A uniformly random $\tau=0.1$ fraction of measurements is flipped, and we compare unnormalized \BIHT{} with the per-iterate normalized variant. Boxplots show the directional error over $100$ independent trials with $n=2000$, $s=5$, $m=1000$, and $60$ iterations. The dashed curve is the constants-suppressed rate guide from \cref{thm:adversarial-positive}, and the dotted vertical line marks the corresponding stopping time. 
        A uniformly random $\tau=0.1$ fraction of the measurements is flipped: BIHT with and without per-iterate normalization. Here the rate guide ends at the median error of normalized BIHT, and the dotted line marks the corresponding stopping time of \cref{cor:adversarial-deterministic-stopping}.
        }
        \label{fig:numerics-corruptions}
    \end{subfigure}

    \caption{
    Numerical experiments illustrating the two regimes studied in the paper: finite-time convergence without normalization in the noiseless case, and early recovery under sign corruptions, which cannot be sustained without per-iterate normalization. Each panel shows the error over $100$ independent trials with $n=2000$, $s=5$ and $m=1000$, where the error is the $\ell_2$-distance between the normalized iterate and the target $\xstar=(1,\dots,1,0,\dots,0)/\sqrt{s}$. Each trial draws a new $A$ with independent $\mathcal{N}(0,1)$ entries and an initialization chosen uniformly at random among $s$-sparse unit vectors. The dashed rate guides start at the median error at $t=1$ and halve at every iteration, as in the proofs of \cref{thm:noiseless,thm:adversarial-positive}, until they reach the median error over the last ten iterations.
    }
    \label{fig:numerics}
\end{figure}

\end{document}